\documentclass[sort&compress, nonatbib]{elsarticle}

\usepackage{mysty-min}

\begin{document}

\journal{Logical and Algebraic Methods in Programming}

\begin{frontmatter}
  
\title{A Linear-\!Time--Branching-\!Time Spectrum for Behavioral
  Specification Theories}

\author[x]{Uli Fahrenberg}

\author[l]{Axel Legay}

\address[x]{{\'E}cole polytechnique, Palaiseau, France}

\address[l]{Universit{\'e} catholique de Louvain, Belgium}

\begin{abstract}
  We propose behavioral specification theories for most equivalences
  in the linear-time--branching-time spectrum.  Almost all previous
  work on specification theories focuses on bisimilarity, but there is
  a clear interest in specification theories for other preorders and
  equivalences.  We show that specification theories for preorders
  cannot exist and develop a general scheme which allows us to define
  behavioral specification theories, based on disjunctive modal
  transition systems, for most equivalences in the
  linear-time--branching-time spectrum.

\end{abstract}

\end{frontmatter}

\section{Introduction}

Models and specifications are central objects in theoretical computer
science.  In model-based verification, models of computing systems are
held up against specifications of their behaviors, and methods are
developed to check whether or not a given model satisfies a given
specification.

In recent years, behavioral specification theories have seen some
popularity \cite{DBLP:journals/scp/AcetoFFIP13,
  DBLP:conf/fase/BauerDHLLNW12, DBLP:conf/atva/BenesCK11,
  DBLP:conf/avmfss/Larsen89, DBLP:conf/concur/Larsen90,
  DBLP:conf/lics/LarsenX90,
  DBLP:journals/entcs/Raclet08, DBLP:journals/eatcs/AntonikHLNW08,
  DBLP:conf/acsd/BujtorSV15, DBLP:journals/tecs/BujtorV15,
  DBLP:conf/ictac/CaillaudR12}.  Here, the specification formalism is
an extension of the modeling formalism, so that specifications have an
operational interpretation and models are verified by comparing their
operational behavior against the specification's behavior.  Popular
examples of such specification theories are modal transition
systems~\cite{DBLP:journals/eatcs/AntonikHLNW08,
  DBLP:journals/tecs/BujtorV15, DBLP:conf/avmfss/Larsen89},
disjunctive modal transition systems~\cite{DBLP:conf/acsd/BujtorSV15,
  DBLP:conf/atva/BenesCK11,
  DBLP:conf/lics/LarsenX90}, and acceptance
specifications~\cite{DBLP:journals/entcs/Raclet08,
  DBLP:conf/ictac/CaillaudR12}.  Also relations to contracts and
interfaces have been exposed~\cite{DBLP:conf/fase/BauerDHLLNW12,
  DBLP:journals/fuin/RacletBBCLP11}, as have extensions for real-time
and quantitative specifications and for models with
data~\cite{DBLP:journals/mscs/BauerJLLS12,
  DBLP:journals/scp/BertrandLPR12, DBLP:journals/sttt/DavidLLNTW15,
  DBLP:journals/fmsd/BauerFJLLT13, DBLP:journals/acta/FahrenbergL14}.

Except for the work by Vogler~\etal
in~\cite{DBLP:conf/acsd/BujtorSV15, DBLP:journals/tecs/BujtorV15},
behavioral specification theories have been developed only to
characterize bisimilarity.  While bisimilarity is an important
equivalence relation on models, there are many others which also are
of interest.  Examples include nested and $k$-nested
simulation~\cite{DBLP:journals/iandc/GrooteV92,
  DBLP:journals/iandc/AcetoFGI04}, ready or
$\tfrac23$-simulation~\cite{DBLP:conf/popl/LarsenS89}, trace
equivalence~\cite{DBLP:journals/cacm/Hoare78}, impossible
futures~\cite{DBLP:books/sp/Vogler92}, or the failure semantics
of~\cite{DBLP:conf/acsd/BujtorSV15, DBLP:journals/tecs/BujtorV15,
  DBLP:journals/jacm/BrookesHR84, DBLP:journals/acta/Vogler89,
  DBLP:conf/icalp/Pnueli85} and others.

In order to initiate a systematic study of specification theories for
different semantics, we exhibit in this paper specification theories
for most of the equivalences in van~Glabbeek's
linear-time--branching-time spectrum~\cite{inbook/hpa/Glabbeek01}, see
Figure~\ref{fi:spectrum}.

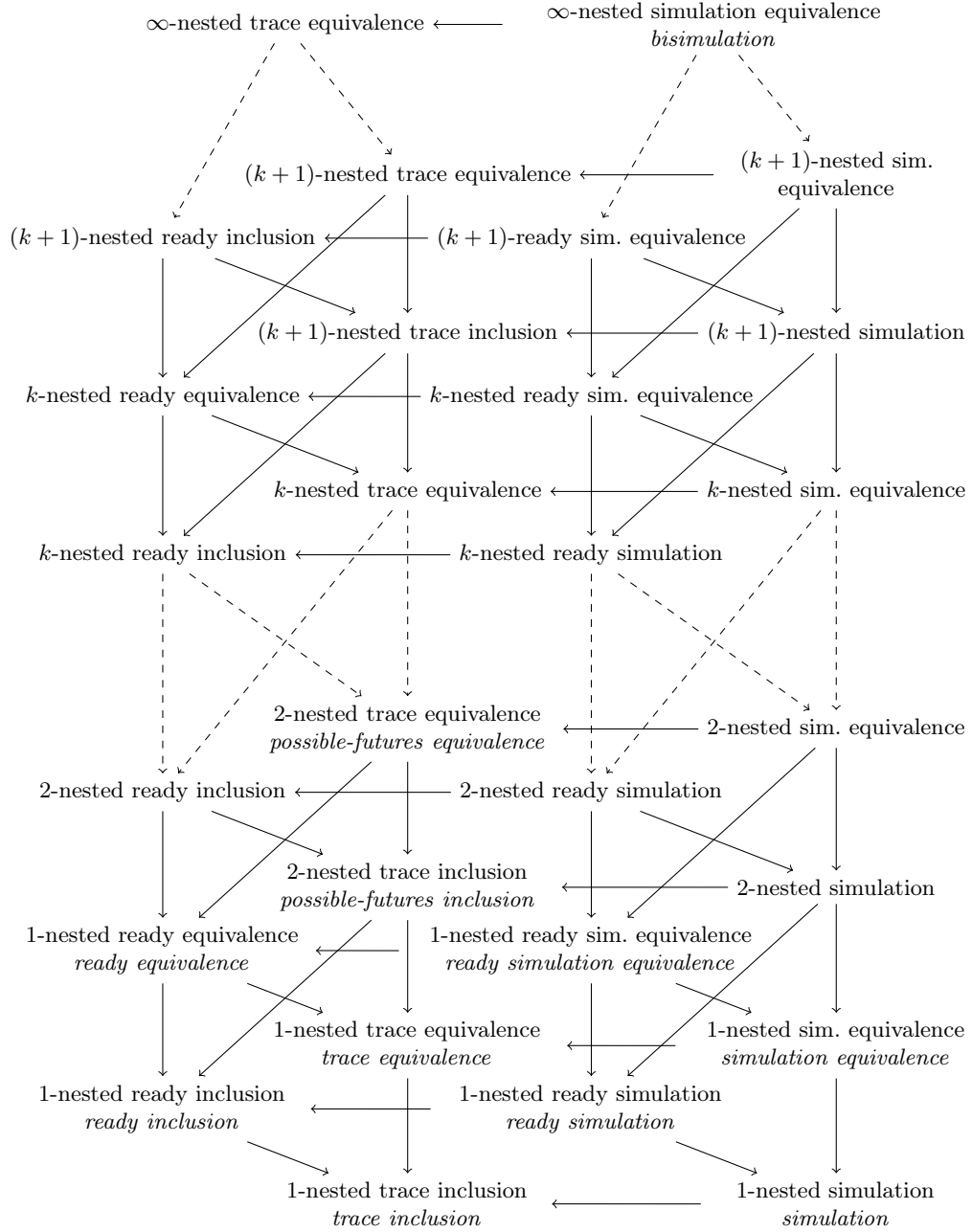
\begin{figure}[p]
  \centering
  \begin{tikzpicture}[->,xscale=.85,yscale=1.1]
    \tikzstyle{every node}=[font=\small,text badly centered]
    \node (traceeq) at (0,.3) {$\infty$-nested trace equivalence};
    \node (k+1-r-trace) at (-2,-2.4) {$( k+ 1)$-nested ready
      inclusion};
    \node (k+1-traceeq) at (2,-1.6) {$( k+ 1)$-nested trace equivalence};
    \node (k-r-traceeq) at (-2,-4.4) {$k$-nested ready equivalence};
    \node (k+1-trace) at (2,-3.6) {$( k+ 1)$-nested trace inclusion};
    \node (k-r-trace) at (-2,-6.4) {$k$-nested ready inclusion};
    \node (k-traceeq) at (2,-5.6) {$k$-nested trace equivalence};
    \node (2-r-trace) at (-2,-9.4) {$2$-nested ready inclusion};
    \node [text width=11.6em] (2-traceeq) at (2,-8.6) {$2$-nested trace
      equivalence \\ \emph{possible-futures equivalence}};
    \node [text width=11.4em] (1-r-traceeq) at (-2,-11.4) {$1$-nested ready
      equivalence \\ \emph{ready equivalence}};
    \node [text width=11.5em] (2-trace) at (2,-10.6) {$2$-nested trace
      inclusion \\ \emph{possible-futures inclusion}};
    \node [text width=11em] (1-r-trace) at (-2,-13.4) {$1$-nested ready
      inclusion \\ \emph{ready inclusion}};
    \node [text width=11.9em] (1-traceeq) at (2,-12.6) {$1$-nested trace
      equivalence \\ \emph{trace equivalence}};
    \node [text width=10.7em] (1-trace) at (2,-14.6) {$1$-nested trace
      inclusion \\ \emph{trace inclusion}};
    \node [text width=16em] (bisim) at (7,.3) {$\infty$-nested
      simulation equivalence \\ \emph{bisimulation}};
    \node (k+1-r-sim) at (5,-2.4) {$( k+ 1)$-ready sim.~equivalence};
    \node [text width=9em] (k+1-simeq) at (9,-1.6) {$( k+ 1)$-nested
      sim. equivalence};
    \node (k-r-simeq) at (5,-4.4) {$k$-nested ready sim.~equivalence};
    \node (k+1-sim) at (9,-3.6) {$( k+ 1)$-nested simulation};
    \node (k-r-sim) at (5,-6.4) {$k$-nested ready simulation};
    \node (k-simeq) at (9,-5.6) {$k$-nested sim.~equivalence};
    \node (2-r-sim) at (5,-9.4) {$2$-nested ready simulation};
    \node (2-simeq) at (9,-8.6) {$2$-nested sim.~equivalence};
    \node [text width=14.5em] (1-r-simeq) at (5,-11.4) {$1$-nested ready
      sim.~equivalence \\ \emph{ready simulation equivalence}};
    \node (2-sim) at (9,-10.6) {$2$-nested simulation};
    \node [text width=12em] (1-r-sim) at (5,-13.4) {$1$-nested ready
      simulation \\ \emph{ready simulation}};
    \node [text width=12em] (1-simeq) at (9,-12.6) {$1$-nested
      sim.~equivalence \\ \emph{simulation equivalence}};
    \node [text width=10em] (1-sim) at (9,-14.6) {$1$-nested
      simulation \\ \emph{simulation}};
    %
    \path (bisim) edge (traceeq);
    \path (k+1-r-sim) edge (k+1-r-trace);
    \path (k+1-simeq) edge (k+1-traceeq);
    \path (k-r-simeq) edge (k-r-traceeq);
    \path (k+1-sim) edge (k+1-trace);
    \path (k-r-sim) edge (k-r-trace);
    \path (k-simeq) edge (k-traceeq);
    \path (2-simeq) edge (2-traceeq);
    \path (2-r-sim) edge (2-r-trace);
    \path (2-sim) edge (2-trace);
    \path (1-r-simeq) edge (1-r-traceeq);
    \path (1-simeq) edge (1-traceeq);
    \path (1-r-sim) edge (1-r-trace);
    \path (1-sim) edge (1-trace);
    \path [dashed] (traceeq) edge (k+1-r-trace);
    \path [dashed] (traceeq) edge (k+1-traceeq);
    \path (k+1-r-trace) edge (k-r-traceeq);
    \path (k+1-r-trace) edge (k+1-trace);
    \path (k+1-traceeq) edge (k-r-traceeq);
    \path (k+1-traceeq) edge (k+1-trace);
    \path (k-r-traceeq) edge (k-r-trace);
    \path (k-r-traceeq) edge (k-traceeq);
    \path (k+1-trace) edge (k-r-trace);
    \path (k+1-trace) edge (k-traceeq);
    \path [dashed] (k-r-trace) edge (2-r-trace);
    \path [dashed] (k-r-trace) edge (2-traceeq);
    \path [dashed] (k-traceeq) edge (2-r-trace);
    \path [dashed] (k-traceeq) edge (2-traceeq);
    \path (2-r-trace) edge (1-r-traceeq);
    \path (2-r-trace) edge (2-trace);
    \path (2-traceeq) edge (1-r-traceeq);
    \path (2-traceeq) edge (2-trace);
    \path (1-r-traceeq) edge (1-r-trace);
    \path (1-r-traceeq) edge (1-traceeq);
    \path (2-trace) edge (1-r-trace);
    \path (2-trace) edge (1-traceeq);
    \path (1-r-trace) edge (1-trace);
    \path (1-traceeq) edge (1-trace);
    \path [dashed] (bisim) edge (k+1-r-sim);
    \path [dashed] (bisim) edge (k+1-simeq);
    \path (k+1-r-sim) edge (k-r-simeq);
    \path (k+1-r-sim) edge (k+1-sim);
    \path (k+1-simeq) edge (k-r-simeq);
    \path (k+1-simeq) edge (k+1-sim);
    \path (k-r-simeq) edge (k-r-sim);
    \path (k-r-simeq) edge (k-simeq);
    \path (k+1-sim) edge (k-r-sim);
    \path (k+1-sim) edge (k-simeq);
    \path [dashed] (k-r-sim) edge (2-r-sim);
    \path [dashed] (k-r-sim) edge (2-simeq);
    \path [dashed] (k-simeq) edge (2-r-sim);
    \path [dashed] (k-simeq) edge (2-simeq);
    \path (2-r-sim) edge (1-r-simeq);
    \path (2-r-sim) edge (2-sim);
    \path (2-simeq) edge (1-r-simeq);
    \path (2-simeq) edge (2-sim);
    \path (1-r-simeq) edge (1-r-sim);
    \path (1-r-simeq) edge (1-simeq);
    \path (2-sim) edge (1-r-sim);
    \path (2-sim) edge (1-simeq);
    \path (1-r-sim) edge (1-sim);
    \path (1-simeq) edge (1-sim);
  \end{tikzpicture}
  \caption{\label{fi:spectrum}%
    The linear-time--branching-time spectrum.  The nodes are different
    preorders and equivalences, and an edge $R_1\longrightarrow R_2$
    or $R_1\dashrightarrow R_2$ indicates that $R_1$ implies $R_2$ and
    that they are inequivalent in general.}
\end{figure}

To develop our systemization, we first have to clarify what precisely
is meant by a specification theory.  This is similar to the attempt at
a uniform framework of specifications
in~\cite{DBLP:conf/fase/BauerDHLLNW12}, but our focus is more general.
Inspired by the seminal work of
Pnueli~\cite{DBLP:conf/icalp/Pnueli85},
Larsen~\cite{DBLP:conf/concur/Larsen90}, and Hennessy and
Milner~\cite{DBLP:journals/jacm/HennessyM85}, we develop the point of
view that a behavioral specification theory is an expressive
specification formalism equipped with a mapping from models to their
characteristic formulae and with a refinement preorder which
generalizes the satisfaction relation between models and
specifications.

We then introduce a general scheme of linear and branching relation
families and show that variants of these characterize most of the
preorders and equivalences in the linear-time--branching-time spectrum
(notably also all of the ones mentioned above).  We transfer our
scheme to disjunctive modal transition systems and use it to define a
linear-time--branching-time spectrum of refinement preorders, each
giving rise to a specification theory for a different equivalence in
the linear-time--branching-time spectrum.

Specification theories as we define them here are useful for
incremental design and verification, as specifications can be refined
until a sufficient level of detail is reached.  The specification
theories developed for bisimilarity
in~\cite{DBLP:journals/scp/AcetoFFIP13, DBLP:conf/atva/BenesCK11,
  DBLP:conf/avmfss/Larsen89, DBLP:conf/concur/Larsen90,
  DBLP:conf/lics/LarsenX90, DBLP:journals/entcs/Raclet08,
  DBLP:journals/eatcs/AntonikHLNW08,
  DBLP:conf/ictac/CaillaudR12} also include operations of conjunction
and composition, hence allowing for compositional design and
verification.  What we present here is a first fundamental study of
specification theories for equivalences other than bisimilarity, and
we leave compositionality for future work.

To sum up, the contributions of this paper are as follows:
\begin{itemize}
\item a clarification of the basic theory of behavioral specification
  theories;
\item a uniform treatment of most of the relations in the
  linear-time--branching-time spectrum;
\item a uniform linear-time--branching-time spectrum of specification
  theories.
\end{itemize}
This article is a revised and extended version of the
paper~\cite{DBLP:conf/sofsem/FahrenbergL17} which has been presented
at the 43rd International Conference on Current Trends in Theory and
Practice of Computer Science (SOFSEM 2017) in Limerick, Ireland.
Compared to~\cite{DBLP:conf/sofsem/FahrenbergL17}, and in addition to
numerous small changes and improvements, motivation and examples,
proofs of all results, as well as two additional sections on a
game-based setting have been added to the paper.

\section{Specification Theories}

We start this paper by introducing and clarifying some concepts
related to models and specifications
from~\cite{DBLP:conf/concur/Larsen90, DBLP:conf/icalp/Pnueli85,
  DBLP:journals/jacm/HennessyM85}.  Let $\Proc$ be a set of models.

\begin{definition}
  A \emph{specification formalism} for $\Proc$ is a structure
  $( \Spec, \mathord{ \models})$, where $\Spec$ is a set of
  specifications and $\mathord{ \models}\subseteq \Proc\times \Spec$
  is the satisfaction relation.
\end{definition}

The models in $\Proc$ serve to represent computing systems, and the
specifications in $\Spec$ represent properties of such systems.  The
\emph{model checking} problem is, given $\mcal I\in \Proc$ and
$\mcal S\in \Spec$, to decide whether $\mcal I\models \mcal S$.

\begin{definition}
  For $\mcal S\in \Spec$,
  $\Mod{ \mcal S}=\{ \mcal I\in \Proc\mid \mcal I\models \mcal S\}$
  denotes its set of \emph{implementations}.
\end{definition}

That is, $\Mod{ \mcal S}$ is the set of models which adhere to the
specification $\mcal S$.  Note that $\models$ and $\Modnull$ are
inter-definable: for $\mcal I\in \Proc$ and $\mcal S\in \Spec$,
$\mcal I\models \mcal S$ iff $\mcal I\in \Mod{ \mcal S}$.

\begin{definition}
  For $\mcal S_1, \mcal S_2\in \Spec$,
  \begin{itemize}
  \item $\mcal S_1$ is \emph{semantically refined} by $\mcal S_2$,
    denoted $\mcal S_1\preceq \mcal S_2$, if
    $\Mod{ \mcal S_1}\subseteq \Mod{ \mcal S_2}$;
  \item $\mcal S_1$ is \emph{semantically equivalent} to $\mcal S_2$,
    denoted $\mcal S_1\approxeq \mcal S_2$, if
    $\Mod{ \mcal S_1}= \Mod{ \mcal S_2}$.
  \end{itemize}
\end{definition}

Hence $\mcal S_1\preceq \mcal S_2$ iff every implementation of
$\mcal S_1$ is also an implementation of $\mcal S_2$, that is, if it
holds for every model that once it satisfies $\mcal S_1$, it
automatically also satisfies $\mcal S_2$.

\begin{definition}
  For $\mcal I\in \Proc$,
  $\Th{ \mcal I}=\{ \mcal S\in \Spec\mid \mcal I\models \mcal S\}$
  denotes its set of \emph{theories}.
\end{definition}

That is, $\Th{ \mcal I}$ is the set of all specifications which are
satisfied by $\mcal I$.  Again, $\models$ and $\Thnull$ are
inter-definable: for $\mcal I\in \Proc$ and $\mcal S\in \Spec$, $\mcal
I\models \mcal S$ iff $\mcal S\in \Th{ \mcal I}$.

As~\cite{DBLP:conf/concur/Larsen90} notes,
the functions $\Modnull: \Spec\to 2^\Proc$ and
$\Thnull: \Proc\to 2^\Spec$ can be extended to functions on sets of
specifications and models by
$\Mod{ A}= \bigcap_{ \mcal S\in A} \Mod{ \mcal S}$ and
$\Th{ B}= \bigcap_{ \mcal I\in A} \Th{ \mcal I}$, and then
$\Modnull: 2^\Spec\rightleftarrows 2^\Proc: \Thnull$ forms a Galois
connection.

\begin{definition}
  For $\mcal I_1, \mcal I_2\in \Proc$,
  \begin{itemize}
  \item $\mcal I_1$ is \emph{behaviorally refined} by $\mcal I_2$,
    denoted $\mcal I_1\sqsubseteq \mcal I_2$, if
    $\Th{ \mcal I_1}\subseteq \Th{ \mcal I_2}$;
  \item $\mcal I_1$ is \emph{behaviorally equivalent} to $\mcal I_2$,
    denoted $\mcal I_1\boxeq \mcal I_2$, if
    $\Th{ \mcal I_1}= \Th{ \mcal I_2}$;
  \end{itemize}
\end{definition}

Hence $\mcal I_1\boxeq \mcal I_2$ iff $\mcal I_1$ and $\mcal I_2$
satisfy precisely the same specifications.

In terminology first introduced
in~\cite{DBLP:journals/jacm/HennessyM85}, the specification formalism
$( \Spec, \mathord{ \models})$ is said to be \emph{adequate} for
$\boxeq$.  In fact, the usual point of view is sightly different:
normally, $\Proc$ comes equipped with some equivalence relation
$\sim$, and then one says that $( \Spec, \mathord{ \models})$ is
adequate for $( \Proc, \mathord{ \sim})$ if
$\mathord{ \boxeq}= \mathord{ \sim}$.  It is clear that $\sim$ is not
needed to reason about specification formalisms; we can simply declare
that $( \Spec, \mathord{ \models})$ is adequate for whatever model
equivalence $\boxeq$ it \emph{induces}.

\begin{definition}
  A specification $\mcal S\in \Spec$ is a \emph{characteristic
    formula} for $\mcal I\in \Proc$ if $\mcal I\models \mcal S$ and
  for all $\mcal I'\models \mcal S$, $\mcal I'\boxeq \mcal I$.
\end{definition}

This was introduced in~\cite{DBLP:conf/icalp/Pnueli85}.  We record the
following property which follows directly from the definitions:

\begin{lemma}
  \label{le:char-Th}
  A specification $\mcal S\in \Spec$ is a characteristic formula for
  $\mcal I\in \Proc$ iff it holds for all $\mcal I'\in \Proc$ that
  $\mcal S\in \Th{ \mcal I'}$ iff $\Th{ \mcal I}= \Th{ \mcal
    I'}$. \qed
\end{lemma}

Not surprisingly, characteristic formulae are unique up to semantic
equivalence:

\begin{lemma}
  If $\mcal S_1$ and $\mcal S_2$ are characteristic formulae for
  $\mcal I\in \Proc$, then $\mcal S_1\approxeq \mcal S_2$.
\end{lemma}

\begin{proof}
  By Lemma~\ref{le:char-Th}, it holds for all $\mcal I'\in \Proc$ that
  $\mcal I'\models \mcal S_1$ iff $\Th{ \mcal I}= \Th{ \mcal I'}$, iff
  $\mcal I'\models \mcal S_2$. \qed
\end{proof}

Again following~\cite{DBLP:conf/icalp/Pnueli85}, the specification
formalism $( \Spec, \mathord{ \models})$ is said to be
\emph{expressive} for $\Proc$ if every $\mcal I\in \Proc$ admits a
characteristic formula.  Our first result seems to have been
overlooked in~\cite{DBLP:conf/concur/Larsen90,
  DBLP:conf/icalp/Pnueli85, DBLP:journals/jacm/HennessyM85}: in an
expressive specification formalism, the preorder $\sqsubseteq$ is, in
fact, an equivalence.

\begin{proposition}
  \label{le:expr->symm}
  If $\Spec$ is expressive for $\Proc$, then
  $\mathord{ \sqsubseteq}= \mathord{ \boxeq}$.
\end{proposition}

\begin{proof}
  Let $\mcal I_1, \mcal I_2\in \Proc$ and assume
  $\mcal I_1\sqsubseteq \mcal I_2$.  Let $\mcal S_1\in \Spec$ be a
  characteristic formula for $\mcal I_1$, then
  $\mcal S_1\in \Th{ \mcal I_1}$.  But
  $\Th{ \mcal I_1}\subseteq \Th{ \mcal I_2}$, hence
  $\mcal S_1\in \Th{ \mcal I_2}$.  By Lemma~\ref{le:char-Th}, this
  implies $\mcal I_2\boxeq \mcal I_1$. \qed
\end{proof}

\begin{example}
  A very simple specification formalism is $\Spec= 2^\Proc$, that is,
  specifications are sets of models.  In that case,
  $\mathord{ \models}= \mathord{ \in}$ is the element-of relation, and
  $\Mod{ \mcal S}= \mcal S$, thus $\mcal S_1\preceq \mcal S_2$ iff
  $\mcal S_1\subseteq \mcal S_2$ and $\mcal S_1\approxeq \mcal S_2$
  iff $\mcal S_1= \mcal S_2$.

  Every $\mcal I\in \Proc$ has characteristic formula
  $\{ \mcal I\}\in \Spec$, hence $2^\Proc$ is expressive for $\Proc$,
  so that $\mathord{ \sqsubseteq}= \mathord{ \boxeq}$.  Further, if
  $\mcal I_1\boxeq \mcal I_2$, then $\mcal I_2\in\{ \mcal I_1\}$,
  hence $\mcal I_1= \mcal I_2$.  We have shown that $2^\Proc$ is
  adequate for equality~$=$. \qed
\end{example}

\begin{example}
  \emph{Hennessy-Milner logic}~\cite{DBLP:journals/jacm/HennessyM85}
  is a well-known specification formalism for labeled transition
  systems (see Definition~\ref{de:lts} of $\LTS$ below).  It consists of
  formulae generated by the abstract syntax
\begin{equation*}
  \HML\ni \phi, \psi\Coloneqq \ltrue\mid \lfalse\mid \phi\land
  \psi\mid \phi\lor \psi\mid \langle a\rangle \phi\mid[ a] \phi \quad(
  a\in \Sigma)\,,
\end{equation*}
with semantics defined by $\Mod \ltrue= \LTS$,
$\Mod \lfalse= \emptyset$,
$\Mod{ \phi\land \psi}= \Mod \phi\cap \Mod \psi$,
$\Mod{ \phi\lor \psi}= \Mod \phi\cup \Mod \psi$, and
\begin{align*}
  &\Mod{ \langle a\rangle \phi}=\{( S, s^0, T)\in \LTS\mid \exists(
  s^0, a, s)\in T:( S, s, T)\in \Mod \phi\}\,, \\
  &\Mod{ [ a] \phi}=\{( S, s^0, T)\in \LTS\mid \forall( s^0, a,
  s)\in T:( S, s, T)\in \Mod \phi\}\,.
\end{align*}

$\HML$ admits a semantic form of negation, \emph{complementation},
which is defined inductively by $\ltrue^c= \lfalse$,
$\lfalse^c= \ltrue$, $( \phi\land \psi)^c= \phi^c\lor \psi^c$,
$( \phi\lor \psi)^c= \phi^c\land \psi^c$,
$( \langle a\rangle \phi)^c=[ a] \phi^c$, and
$([ a] \phi)^c= \langle a\rangle \phi^c$.  It can be
shown~\cite{books/AcetoILS07} that for all $\phi\in \HML$,
$\Mod{ \phi^c}= \Proc\setminus \Mod \phi$.

Now let $\mcal I_1, \mcal I_2\in \LTS$ and assume
$\mcal I_1\sqsubseteq \mcal I_2$, then it holds for all $\phi\in \HML$
that $\mcal I_1\models \phi$ implies $\mcal I_2\models \phi$.  By
contraposition, $\mcal I_2\models \phi^c$ implies
$\mcal I_1\models \phi^c$ for all $\phi\in \HML$, so that
$\mcal I_2\sqsubseteq \mcal I_1$.  We have shown that
$\mathord{ \sqsubseteq}= \mathord{ \boxeq}$.  In fact, by the
Hennessy-Milner theorem~\cite{DBLP:journals/jacm/HennessyM85},
$\boxeq$ is bisimilarity, so that $\HML$ is adequate for bisimilarity.

Even though $\mathord{ \sqsubseteq}= \mathord{ \boxeq}$, it can be
shown~\cite{books/AcetoILS07} that $\HML$ is \emph{not}
expressive. \qed
\end{example}

\section{Behavioral Specification Theories}

We are ready to introduce what we mean by a behavioral specification
theory: an expressive specification formalism with extra structure.
This mainly sums up and clarifies ideas already present
in~\cite{DBLP:conf/concur/Larsen90, DBLP:conf/fase/BauerDHLLNW12}, but
we make a connection between specification theories and characteristic
formulae which is new.  Specifically, we will see that a central
ingredient in a specification theory is a function $\chi$ which maps
models to their characteristic formulae.

\begin{definition}
  A (behavioral) \emph{specification theory} for $\Proc$ is
  a~specification formalism $( \Spec, \mathord{ \models})$ for $\Proc$
  together with a mapping $\chi: \Proc\to \Spec$ and a~preorder $\le$
  on $\Spec$, called \emph{modal refinement}, subject to the following
  conditions:
  \begin{itemize}
  \item for every $\mcal I\in \Proc$, $\chi( \mcal I)$ is
    a characteristic formula for $\mcal I$;
  \item for all $\mcal I\in \Proc$ and all $\mcal S\in \Spec$, $\mcal
    I\models \mcal S$ iff $\chi( \mcal I)\le \mcal S$.
  \end{itemize}
\end{definition}

The equivalence relation
$\mathord{ \equiv}= \mathord{ \le}\cap \mathord{ \ge}$ on $\Spec$ is
called \emph{modal equivalence}.  Note that specification theories are
indeed expressive; also, $\models$ is fully determined by~$\le$.

In a categorical sense, the function $\chi: \Proc\to \Spec$ is a
\emph{section} of the Galois connection
$\Modnull: 2^\Spec\rightleftarrows 2^\Proc: \Thnull$.  Indeed, we have
$\chi( \mcal I)\in \Th{ \mcal I}$ for all $\mcal I\in \Proc$ and
$\mcal I'\boxeq \mcal I$ for all $\mcal I'\in \Mod{ \chi( \mcal I)}$,
and these properties are characterizing for $\chi$.  Further,
$\Th{ \mcal I}=\{ \mcal S\mid \chi( \mcal I)\le \mcal S\}= \chi( \mcal
I)\mathord\uparrow$ is the \emph{upward closure} of $\chi( \mcal I)$.

We sum up a few consequences of the definition: modal refinement
(equivalence, resp.) implies semantic refinement (equivalence, resp.),
and on characteristic formulae, all refinements and equivalences
collapse.

\begin{proposition}
  \label{pr:spec-basic}
  Let $( \Spec, \chi, \mathord{ \le})$ be a specification theory for
  $\Proc$.
  \begin{enumerate}
  \item For all $\mcal S_1, \mcal S_2\in \Spec$,
    $\mcal S_1\le \mcal S_2$ implies $\mcal S_1\preceq \mcal S_2$ and
    $\mcal S_1\equiv \mcal S_2$ implies
    $\mcal S_1\approxeq \mcal S_2$.
  \item For all $\mcal I_1, \mcal I_2\in \Proc$, the following are
    equivalent: $\chi( \mcal I_1)\le \chi( \mcal I_2)$,
    $\chi( \mcal I_2)\le \chi( \mcal I_1)$,
    $\chi( \mcal I_1)\preceq \chi( \mcal I_2)$,
    $\chi( \mcal I_2)\preceq \chi( \mcal I_1)$,
    $\mcal I_1\boxeq \mcal I_2$.
  \end{enumerate}
\end{proposition}

\begin{proof}
  The first claim follows from transitivity of $\le$: if
  $\mcal I\in \Mod{ \mcal S_1}$, then
  $\chi( \mcal I)\le \mcal S_1\le \mcal S_2$, hence
  $\chi( \mcal I)\le \mcal S_2$, thus $\mcal I\in \Mod{ \mcal S_2}$.

  For the second claim, let $\mcal I_1, \mcal I_2\in \Proc$.
  \begin{itemize}
  \item If $\chi( \mcal I_1)\le \chi( \mcal I_2)$, then
    $\chi( \mcal I_1)\preceq \chi( \mcal I_2)$ by the first part.
  \item If $\chi( \mcal I_1)\preceq \chi( \mcal I_2)$, then
    $\Mod{ \chi( \mcal I_1)}\subseteq \Mod{ \chi( \mcal I_2)}$.  But
    $\mcal I_1\in \Mod{ \chi( \mcal I_1)}$, hence
    $\mcal I_1\in \Mod{ \chi( \mcal I_2)}$, which, as
    $\chi( \mcal I_2)$ is characteristic, implies
    $\mcal I_1\boxeq \mcal I_2$.  Also,
    $\mcal I_1\in \Mod{ \chi( \mcal I_2)}$ implies
    $\chi( \mcal I_1)\le \chi( \mcal I_2)$.
  \item Assume $\mcal I_1\boxeq \mcal I_2$ and let
    $\mcal I\in \Mod{ \chi( \mcal I_1)}$.  Then
    $\mcal I\boxeq \mcal I_1$, hence $\mcal I\boxeq \mcal I_2$, which
    implies $\mcal I\in \Mod{ \chi( \mcal I_2)}$.  We have shown that
    $\chi( \mcal I_1)\preceq \chi( \mcal I_2)$.
  \end{itemize}

  We have shown that $\chi( \mcal I_1)\le \chi( \mcal I_2)$ iff
  $\chi( \mcal I_1)\preceq \chi( \mcal I_2)$ iff
  $\mcal I_1\boxeq \mcal I_2$, and reversing the roles of $\mcal I_1$
  and $\mcal I_2$ gives the other equivalences.  \qed
\end{proof}

The second part of the proposition means that the mapping
$\chi: \Proc\to \Spec$ is an \emph{embedding up to equivalence}: for
all $\mcal I_1, \mcal I_2\in \Proc$, $\mcal I_1\boxeq \mcal I_2$ iff
$\chi( \mcal I_1)\equiv \chi( \mcal I_2)$ iff
$\chi( \mcal I_1)\approxeq \chi( \mcal I_2)$.  Because of this, most
work in specification theories \emph{identifies} models $\mcal I$ with
their characteristic formulae $\chi( \mcal I)$; for reasons of
clarity, we will not make this identification here.

We finish this section with a lemma which shows that the property of
$\chi( \mcal I)$ being characteristic formulae follows when $\le$ is
symmetric on models.

\begin{lemma}
  \label{le:specth}
  Let $\Spec$ be a set, $\chi: \Proc\to \Spec$ a mapping and
  $\mathord{ \le}\subseteq \Spec\times \Spec$ a preorder.  If the
  restriction of $\le$ to the image of $\chi$ is symmetric, then $(
  \Spec, \chi, \mathord{ \le})$ is a specification theory for
  $\Proc$.
\end{lemma}

\begin{proof}
We know that $\chi( \mcal I_1)\le \chi( \mcal I_2)$ iff
$\chi( \mcal I_2)\le \chi( \mcal I_1)$ for all
$\mcal I_1, \mcal I_2\in \Proc$.  Let $\mcal I\in \Proc$; we need to
show that $\chi( \mcal I)$ is a characteristic formula for $\mcal I$.

First, by reflexivity of $\le$, $\chi( \mcal I)\le \chi( \mcal I)$
implies $\mcal I\models \chi( \mcal I)$.  Now let $\mcal I'\in \Proc$
and assume $\mcal I'\models \chi( \mcal I)$, that is,
$\chi( \mcal I')\le \chi( \mcal I)$.  We show that
$\Th{ \mcal I'}\supseteq \Th{ \mcal I}$.  Let
$\mcal S\in \Th{ \mcal I}$, then $\mcal I\models \mcal S$, that is,
$\chi( \mcal I)\le \mcal S$.  But $\le$ is transitive, so
$\chi( \mcal I')\le \chi( \mcal I)\le \mcal S$ implies
$\chi( \mcal I')\le \mcal S$.  Hence $\mcal I'\models \mcal S$, so
that $\mcal S\in \Th{ \mcal I'}$.

We have shown that $\chi( \mcal I')\le \chi( \mcal I)$ implies
$\Th{ \mcal I'}\supseteq \Th{ \mcal I}$.  By symmetry of $\le$ on the
image of $\chi$, $\chi( \mcal I')\le \chi( \mcal I)$ implies
$\chi( \mcal I)\le \chi( \mcal I')$, which in turn implies
$\Th{ \mcal I}\supseteq \Th{ \mcal I'}$.  We have proven that
$\mcal I'\models \chi( \mcal I)$ implies
$\Th{ \mcal I'}= \Th{ \mcal I}$. \qed
\end{proof}

\begin{example}
  We have seen that Hennessy-Milner logic is not expressive, hence
  $\HML$ cannot serve as basis for a specification theory for $\LTS$.
  The standard remedy for expressivity is to add recursion to the
  logic, see~\cite{DBLP:journals/tcs/Larsen90, books/AcetoILS07}; we
  will in Sect.~\ref{se:dmts} below expose a specification theory
  based on Hennessy-Milner logic with recursion and maximal fixed
  points.

  For our other example, $\Spec= 2^\Proc$, we can let
  $\chi( \mcal I)=\{ \mcal I\}$ and
  $\mathord{ \le}= \mathord{ \subseteq}$.  Then $\mcal I\in \mcal S$
  iff $\{ \mcal I\}\subseteq \mcal S$, \ie~$\mcal I\models \mcal S$
  iff $\chi( \mcal I)\le \mcal S$.  This shows that
  $( 2^\Proc, \chi, \mathord{ \subseteq})$ is a specification theory
  for $\Proc$ (which is adequate and expressive for equality).
\end{example}

\section{Disjunctive Modal Transition Systems}
\label{se:dmts}

We proceed to recall disjunctive modal transition systems and how
these can serve as a specification theory for bisimilarity.  The
material in this section is well-known, but our definitions from the
previous sections allow for much more succinctness, for example in
Proposition~\ref{pr:dmts-bisim} below.

From now on, $\Proc$ will be the set $\LTS$ of (finite) \emph{labeled
  transition systems} over a fixed finite alphabet $\Sigma$:

\begin{definition}
  \label{de:lts}
  A \emph{labeled transition system} $( S, s^0, T)$ consists of a
  finite set of states $S$, an initial state $s^0\in S$, and
  transitions $T\subseteq S\times \Sigma\times S$ labeled with symbols
  from $\Sigma$.
\end{definition}

Recall~\cite{DBLP:journals/tcs/Milner83, DBLP:conf/tcs/Park81} that
two LTS $( S_1, s^0_1, T_1)$ and $( S_2, s^0_2, T_2)$ are
\emph{bisimilar} if there exists a relation $R\subseteq S_1\times S_2$
such that $( s^0_1, s^0_2)\in R$ and for all $( s_1, s_2)\in R$,
\begin{itemize}
\item for all $( s_1, a, t_1)\in T_1$, there is $( s_2, a, t_2)\in
  T_2$ with $( t_1, t_2)\in R$,
\item for all $( s_2, a, t_2)\in T_2$, there is $( s_1, a, t_1)\in
  T_1$ with $( t_1, t_2)\in R$.
\end{itemize}

\begin{definition}
  A \emph{disjunctive modal transition system} (DMTS) is a tuple
  $\mcal D=( S, S^0, \omay, \omust)$ consisting of finite sets
  $S\supseteq S^0$ of states and initial states, a
  \emph{may}-transition relation
  $\omay\subseteq S\times \Sigma\times S$, and a \emph{disjunctive
    must}-transition relation
  $\omust\subseteq S\times 2^{ \Sigma\times S}$.  It is assumed that
  for all $( s, N)\in \omust$ and all $( a, t)\in N$,
  $( s, a, t)\in \omay$.
\end{definition}

DMTS were introduced in~\cite{DBLP:conf/lics/LarsenX90}, but note that
we permit several (or no) initial states here.  The set of DMTS is
denoted $\DMTS$.

As customary, we write $s\may a t$ instead of $( s, a, t)\in \omay$
and $s\must{} N$ instead of $( s, N)\in \omust$.  The intuition is
that may-transitions $s\may a t$ specify which transitions are
permitted in an implementation, whereas a~must-transition $s\must{} N$
stipulates a disjunctive requirement: at least one of the choices
$( a, t)\in N$ has to be implemented.

\begin{definition}
  A \emph{modal refinement} of two DMTS
  $\mcal D_1=( S_1, S^0_1, \omay_1, \omust_1)$,
  $\mcal D_2=( S_2, S^0_2, \omay_2, \omust_2)$ is a relation
  $R\subseteq S_1\times S_2$ for which it holds of all
  $( s_1, s_2)\in R$ that
\begin{itemize}
\item
  $\forall s_1\may a_1 t_1: \exists s_2\may a_2 t_2:( t_1, t_2)\in R$;
\item
  $\forall s_2\must{}_2 N_2: \exists s_1\must{}_1 N_1: \forall( a,
  t_1)\in N_1: \exists( a, t_2)\in N_2:( t_1, t_2)\in R$;
\end{itemize}
and such that for all $s^0_1\in S^0_1$, there exists $s^0_2\in S^0_2$
for which $( s^0_1, s^0_2)\in R$.
\end{definition}

Let
$\mathord{ \le}\subseteq \DMTS\times \DMTS$ be the relation defined by
$\mcal D_1\le \mcal D_2$ iff there exists a modal refinement as above
(a \emph{witness} for $\mcal D_1\le \mcal D_2$).  Clearly, $\le$ is a
preorder.

LTS are embedded into $\DMTS$ as follows.  For an LTS
$\mcal I=( S, s^0, T)$, let
$\chi( \mcal I)=( S,\{ s^0\}, \omay, \omust)$ be the DMTS with
$\omay= T$ and $\omust=\{( s,\{( a, t)\})\mid( s, a, t)\in T\}$.  The
following proposition reformulates well-known facts about DMTS and
modal refinement.

\begin{proposition}
  \label{pr:dmts-bisim}
  $( \DMTS, \chi, \mathord{ \le})$ is a specification theory for
  $\LTS$ adequate for bisimilarity.
\end{proposition}

\begin{proof}
  In lieu of Lemma~\ref{le:specth}, we show that $\le$ is
  bisimilarity, hence symmetric, on the image of $\chi$.  Let
  $\mcal I_1, \mcal I_2\in \LTS$ and assume
  $\chi( \mcal I_1)\le \chi( \mcal I_2)$.  Write
  $\mcal I_1=( S_1, s^0_1, T_1)$, $\mcal I_2=( S_2, s^0_2, T_2)$,
  $\chi( \mcal I_1)=( S_1,\{ s^0_1\}, \omay_1, \omust_1)$, and
  $\chi( \mcal I_2)=( S_2,\{ s^0_2\}, \omay_2, \omust_2)$.

  We have a relation $R\subseteq S_1\times S_2$ such that
  $( s_1^0, s_2^0)\in R$ and for all $( s_1, s_2)\in R$,
  $\forall s_1\may a_1 t_1: \exists s_2\may a_2 t_2:( t_1, t_2)\in R$
  and
  $\forall s_2\must{}_2 N_2: \exists s_1\must{}_1 N_1: \forall( a,
  t_1)\in N_1: \exists( a, t_2)\in N_2:( t_1, t_2)\in R$.
  Let $( s_1, s_2)\in R$.  We show that $R$ is a bisimulation.

  Let $( s_1, a, t_1)\in T_1$.  Then $s_1\may a_1 t_1$, so that we
  have a transition $s_2\may a_2 t_2$ with $( t_1, t_2)\in R$.  By
  definition of $\chi( \mcal I_1)$, $( s_2, a, t_2)\in T_2$.

  Let $( s_2, a, t_2)\in T_2$.  Then $s_2\must{}_2 N_2=\{( a, t_2)\}$,
  hence there is $s_1\must{}_1 N_1$ such that
  $\forall( a, t_1)\in N_1: \exists( a, t_2')\in N_2:( t_1, t_2')\in
  R$.
  But then $t_2'= t_2$, and by definition of $\chi( \mcal I_2)$,
  $N_1=\{( a, t_1)\}$ must be a one-element set, hence
  $( s_1, a, t_1)\in T_1$ and $( t_1, t_2)\in R$.

  We have shown that $\chi( \mcal I_1)\le \chi( \mcal I_2)$ implies
  that $\mcal I_1$ and $\mcal I_2$ are bisimilar; the proof of the
  other direction is similar.  \qed
\end{proof}

\subsection{Hennessy-Milner Logic with Maximal Fixed Points}

It is shown in~\cite{DBLP:conf/concur/BenesDFKL13,
  DBLP:conf/ictac/FahrenbergLT14} that there is a bijective
translation between DMTS and Hennessy-Milner logic with recursion and
maximal fixed points~\cite{DBLP:journals/tcs/Larsen90}.  For a finite
set $X$ of variables, let $\HML( X)$ be the set of formulae generated
as follows:
\begin{equation*}
  \HML( X)\ni \phi, \psi\Coloneqq \ltrue\mid \lfalse\mid \phi\land
  \psi\mid \phi\lor \psi\mid \langle a\rangle \phi\mid[ a] \phi \mid x
  \quad( a\in \Sigma, x\in X)
\end{equation*}

A \emph{recursive Hennessy-Milner
  formula}~\cite{DBLP:conf/concur/BenesDFKL13,
  DBLP:journals/tcs/Larsen90, DBLP:conf/ictac/FahrenbergLT14} is a
tuple $\mcal H=( X, X^0, \Delta)$ consisting of finite sets
$X\supseteq X^0$ of variables and initial variables and a
\emph{declaration} $\Delta: X\to \HML( X)$.  The set of such formulae
is denoted $\HMLR$.  The semantics of a formula $\mcal H\in \HMLR$ is
a set $\sem{ \mcal H}\in \LTS$ which is defined as a maximal fixed
point, see~\cite{DBLP:conf/concur/BenesDFKL13,
  DBLP:journals/tcs/Larsen90, books/AcetoILS07} for details.


In~\cite{DBLP:conf/concur/BenesDFKL13,
  DBLP:conf/ictac/FahrenbergLT14}, and extending results
of~\cite{DBLP:journals/tcs/BoudolL92, DBLP:conf/avmfss/Larsen89}, it
is shown that there is a bijective translation between DMTS and
recursive HML formulae.  That is, there are mappings
$\dh: \DMTS\to \HMLR$ and $\hd: \HMLR\to \DMTS$ such that
$\hd\circ \dh$ and $\dh\circ \hd$ are identities.

We can now define modal refinement of recursive HML formulae by
$\mcal H_1\le \mcal H_2$ iff
$\hd(\mcal H_1)\le \hd( \mcal H_2)$.  We also embed $\LTS$ into
$\HMLR$ by $\chi( \mcal I)= \dh( \chi_\DMTS( \mcal I))$, where
$\chi_\DMTS$ is the embedding $\LTS\to \DMTS$; this is the usual
characteristic-formula construction from, for example,
\cite{books/AcetoILS07}.

\begin{proposition}[\cite{DBLP:conf/concur/BenesDFKL13,
    DBLP:conf/ictac/FahrenbergLT14}]
  $( \HMLR, \chi, \mathord{ \le})$ is a specification theory for
  $\LTS$ adequate for bisimilarity. \qed
\end{proposition}

\section{A Specification Theory for Simulation Equivalence}
\label{se:simeq-spec}

We want to construct specification theories for other interesting
relations in the linear-time--branching-time
spectrum~\cite{inbook/hpa/Glabbeek01}.  Given
Proposition~\ref{le:expr->symm} and the fact that specification theories are
expressive, we know that it is futile to look for specification
theories for \emph{preorders} in the spectrum.  What we \emph{can} do,
however, is find specification theories for the \emph{equivalences} in
the spectrum.  To warm up, we start out by a specification theory for
simulation equivalence.

Recall~\cite{DBLP:journals/tcs/Larsen87} that a \emph{simulation} of
LTS $( S_1, s^0_1, T_1)$, $( S_2, s^0_2, T_2)$ is a relation
$R\subseteq S_1\times S_2$ such that $( s^0_1, s^0_2)\in R$ and for
all $( s_1, s_2)\in R$,
\begin{itemize}
\item for all $( s_1, a, t_1)\in T_1$, there is $( s_2, a, t_2)\in
  T_2$ with $( t_1, t_2)\in R$.
\end{itemize}
LTS $( S_1, s^0_1, T_1)$ and $( S_2, s^0_2, T_2)$ are said to be
\emph{simulation equivalent} if there exist a simulation
$R^1\subseteq S_1\times S_2$ \emph{and} a simulation
$R^2\subseteq S_2\times S_1$.

\begin{definition}
  \label{de:lesim}
  Let $\mcal D_1=( S_1, S^0_1, \omay_1, \omust_1)$,
  $\mcal D_2=( S_2, S^0_2, \omay_2, \omust_2)$ be DMTS.  A
  \emph{simulation refinement} consists of two relations
  $R_1, R_2\subseteq S_1\times S_2$ such that
  \begin{enumerate}
  \item \label{en:lesim.init}
    $\forall s^0_1\in S^0_1: \exists s^0_2\in S^0_2:( s^0_1, s^0_2)\in
    R_1$
    and
    $\forall s^0_2\in S^0_2: \exists s^0_1\in S^0_1:( s^0_1, s^0_2)\in
    R_2$;
  \item \label{en:lesim.left}
    $\forall( s_1, s_2)\in R_1: \forall s_1\may a_1 t_1: \exists
    s_2\may a_2 t_2:( t_1, t_2)\in R_1$;
  \item \label{en:lesim.right}
    $\forall( s_1, s_2)\in R_2: \forall s_2\must{}_2 N_2: \exists
    s_1\must{}_1 N_1: \forall( a, t_1)\in N_1: \exists( a, t_2)\in
    N_2:( t_1, t_2)\in R_2$.
  \end{enumerate}
\end{definition}

Intuitively, $R_1$ is a simulation of may-transitions from $\mcal D_1$
to $\mcal D_2$, whereas $R_2$ is a simulation of disjunctive
must-transitions from $\mcal D_2$ to $\mcal D_1$.  Let
$\mathord{ \lesim}\subseteq \DMTS\times \DMTS$ be the relation defined
by $\mcal D_1\lesim \mcal D_2$ iff there exists a simulation
refinement as above.  Clearly, $\lesim$ is a preorder.  A direct proof
of the following theorem, similar to the one of
Proposition~\ref{pr:dmts-bisim}, is shown below, but it also follows
from the later
Theorem~\ref{th:spec-game}.

\begin{theorem}
  \label{th:simeq-spec}
  $( \DMTS, \chi, \mathord{ \lesim})$ forms a specification theory for
  $\LTS$ adequate for simulation equivalence.
\end{theorem}

\begin{proof}
We show that $\lesim$ is simulation equivalence, hence symmetric, on
the image of $\chi$ and apply Lemma~\ref{le:specth}.  Let
$\mcal I_1, \mcal I_2\in \LTS$ and assume
$\chi( \mcal I_1)\lesim \chi( \mcal I_2)$.  Write
$\mcal I_1=( S_1, s^0_1, T_1)$, $\mcal I_2=( S_2, s^0_2, T_2)$,
$\chi( \mcal I_1)=( S_1,\{ s^0_1\}, \omay_1, \omust_1)$, and
$\chi( \mcal I_2)=( S_2,\{ s^0_2\}, \omay_2, \omust_2)$.

Let $R_1, R_2\subseteq S_1\times S_2$ be relations as of
Definition~\ref{de:lesim}.  Then $( s^0_1, s^0_2)\in R_1$ and
$( s^0_1, s^0_2)\in R_2$.  We show that $R_1\subseteq S_1\times S_2$
and $R_2^{ -1}\subseteq S_2\times S_1$ are simulations.

Let $( s_1, s_2)\in R_1$ and $( s_1, a, t_1)\in T_1$.  Then
$s_1\may a_1 t_1$, hence there is $s_2\may a_2 t_2$ such that
$( t_1, t_2)\in R_1$.  But then also $( s_2, a, t_2)\in T_2$.

Let $( s_2, s_1)\in R_2^{ -1}$ and $( s_2, a, t_2)\in T_2$.  Then
$s_2\must{}_2 N_2=\{( a, s_2)\}$, hence there is $s_1\must{}_1 N_1$
such that
$\forall( a, t_1)\in N_1: \exists( a, t_2')\in N_2:( t_1, t_2')\in
R_2$.
But then $t_2'= t_2$ and $N_1=\{( a, t_1)\}$, hence
$( s_1, a, t_1)\in T_1$ and $( t_2, t_1)\in R_2^{ -1}$.  \qed
\end{proof}

\section{Specification Theories for Branching Equivalences}
\label{se:bspec}

We proceed to generalize the work in the preceding section and develop
DMTS-based specification theories for all \emph{branching}
equivalences in the linear-time--branching-time spectrum in
Figure~\ref{fi:spectrum}.  Examples of such branching equivalences
include the bisimilarity and simulation equivalence which we have
already seen, but also ready simulation
equivalence~\cite{DBLP:conf/popl/LarsenS89} and nested simulation
equivalence~\cite{DBLP:journals/iandc/GrooteV92,
  DBLP:journals/iandc/AcetoFGI04} are important.  We will treat the
linear part of the spectrum, which includes relations such as trace
equivalence~\cite{DBLP:journals/cacm/Hoare78}, impossible-futures
equivalence~\cite{DBLP:books/sp/Vogler92} or failure
equivalence~\cite{DBLP:conf/acsd/BujtorSV15,
  DBLP:journals/tecs/BujtorV15, DBLP:journals/jacm/BrookesHR84,
  DBLP:journals/acta/Vogler89, DBLP:conf/icalp/Pnueli85}, in the next
section.

We start by laying out a scheme which systematically covers all
branching relations in the spectrum.

\begin{definition}
  \label{de:bkswitch-lts}
  Let $k\ge 0$ and
  $\mcal I_1=( S_1, s^0_1, T_1), \mcal I_2=( S_2, s^0_2, T_2)\in
  \LTS$.
  A \emph{branching $k$-switching relation family} from $\mcal I_1$ to
  $\mcal I_2$ consists of relations
  $R^0\!,\dotsc, R^k\subseteq S_1\times S_2$ such that
  $( s_1^0, s_2^0)\in R^0$ and
  \begin{itemize}
  \item for all \emph{even} $j\in\{ 0,\dotsc, k\}$ and
    $( s_1, s_2)\in R^j$:
    \begin{itemize}
    \item
      $\forall( s_1, a, t_1)\in T_1: \exists( s_2, a, t_2)\in T_2:(
      t_1, t_2)\in R^j$;
    \item if $j< k$, then
      $\forall( s_2, a, t_2)\in T_2: \exists( s_1, a, t_1)\in T_1:(
      t_1, t_2)\in R^{ j+ 1}$;
    \end{itemize}
  \item for all \emph{odd} $j\in\{ 0,\dotsc, k\}$ and
    $( s_1, s_2)\in R^j$:
    \begin{itemize}
    \item
      $\forall( s_2, a, t_2)\in T_2: \exists( s_1, a, t_1)\in T_1:(
      t_1, t_2)\in R^j$;
    \item if $j< k$, then
      $\forall( s_1, a, t_1)\in T_1: \exists( s_2, a, t_2)\in T_2:(
      t_1, t_2)\in R^{ j+ 1}$.
    \end{itemize}
  \end{itemize}
\end{definition}

Clearly, a simulation is the same as a branching $0$-switching
relation family.  Also, a branching $1$-switching relation family is a
\emph{nested simulation}: the initial states are related in $R^0$; any
transition in $\mcal I_1$ from a pair $( s_1, s_2)\in R^0$ has to be
matched recursively in $\mcal I_2$; and at any point in time, the
sense of the matching can switch, in that now transitions in
$\mcal I_2$ from a pair $( s_1, s_2)\in R^1$ have to be matched
recursively by transitions in $\mcal I_1$.  In general, a branching
$k$-switching relation family is a $k$-nested simulation, see
also~\cite[Definition~8.5.2]{DBLP:journals/iandc/GrooteV92} which is
similar to ours.  A branching $\infty$-switching relation family is a
bisimulation: any transition in $\mcal I_1$ has to be matched
recursively by one in $\mcal I_2$ and vice versa.  We refer
to~\cite{DBLP:journals/tcs/FahrenbergL14} for more motivation.

\begin{definition}
  Let $k\ge 0$ and
  $\mcal I_1=( S_1, s^0_1, T_1), \mcal I_2=( S_2, s^0_2, T_2)\in
  \LTS$.
  A \emph{branching $k$-ready relation family} from $\mcal I_1$ to
  $\mcal I_2$ is a branching $k$-switching relation family
  $R^0\!,\dotsc, R^k\subseteq S_1\times S_2$ with the extra property
  that for all $( s_1, s_2)\in R^k$:
  \begin{itemize}
  \item if $k$ is even, then
    $\forall( s_2, a, t_2)\in T_2: \exists( s_1, a, t_1)\in T_1$;
  \item if $k$ is odd, then
    $\forall( s_1, a, t_1)\in T_1: \exists( s_2, a, t_2)\in T_2$.
  \end{itemize}
\end{definition}

Hence a branching $0$-ready relation family is the same as a
\emph{ready simulation}: any transition in $\mcal I_1$ has to be
matched recursively by one in $\mcal I_2$; and at any point in time,
precisely the same actions have to be available in the two states.  A
branching $1$-ready relation family would be a nested ready
simulation, and so on.  Branching $k$-switching and $k$-ready relation
families cover all branching relations in the
linear-time--branching-time spectrum.

Because of Proposition~\ref{le:expr->symm}, we are only interested in
equivalences.  For $k\ge 0$ and $\mcal I_1, \mcal I_2\in \LTS$, we
write $\mcal I_1\sim_k \mcal I_2$ if there exist a branching
$k$-switching relation family from $\mcal I_1$ to $\mcal I_2$ and
another from $\mcal I_2$ to $\mcal I_1$.  We write
$\mcal I_1\sim_k^{ \textup r} \mcal I_2$ if there exist a branching
$k$-ready relation family from $\mcal I_1$ to $\mcal I_2$ and another
from $\mcal I_2$ to $\mcal I_1$.  Then $\sim_0$ is simulation
equivalence, $\sim_1$ is nested simulation equivalence, $\sim_\infty$
is bisimilarity, $\sim_0^{ \textup r}$ is ready simulation
equivalence, etc.

We proceed to devise specification theories for $\LTS$ which are
adequate for~$\sim_k$ and~$\sim_k^{ \textup r}$.

\begin{definition}
  \label{de:bkswitch-dmts}
  Let $k\ge 0$,
  $\mcal D_1=( S_1, S^0_1, \omay_1, \omust_1), \mcal D_2=( S_2, S^0_2,
  \omay_2, \omust_2)\in \DMTS$.
  A \emph{branching $k$-switching relation family} from $\mcal D_1$ to
  $\mcal D_2$ consists of relations
  $R_1^0,\dotsc, R_1^k, R_2^0,\dotsc, R_2^k\subseteq S_1\times S_2$
  such that
  \begin{itemize}
  \item
    $\forall s^0_1\in S^0_1: \exists s^0_2\in S^0_2:( s^0_1, s^0_2)\in
    R_1^0$
    and
    $\forall s^0_2\in S^0_2: \exists s^0_1\in S^0_1:( s^0_1, s^0_2)\in
    R_2^0$;
  \item for all \emph{even} $j\in\{ 0,\dotsc, k\}$ and
    $( s_1, s_2)\in R_1^j$:
    \begin{itemize}
    \item
      $\forall s_1\may a_1 t_1: \exists s_2\may a_2 t_2:( t_1, t_2)\in
      R_1^j$;
    \item if $j< k$, then
      $\forall s_2\must{}_2 N_2: \exists s_1\must{}_1 N_1: \forall( a,
      t_1)\in N_1: \exists( a, t_2)\in N_2:( t_1, t_2)\in R_1^{ j+
        1}$;
    \end{itemize}
  \item for all \emph{odd} $j\in\{ 0,\dotsc, k\}$ and
    $( s_1, s_2)\in R_1^j$:
    \begin{itemize}
    \item
      $\forall s_2\must{}_2 N_2: \exists s_1\must{}_1 N_1: \forall( a,
      t_1)\in N_1: \exists( a, t_2)\in N_2:( t_1, t_2)\in R_1^j$;
    \item if $j< k$, then
      $\forall s_1\may a_1 t_1: \exists s_2\may a_2 t_2:( t_1, t_2)\in
      R_1^{ j+ 1}$;
    \end{itemize}
  \item for all \emph{even} $j\in\{ 0,\dotsc, k\}$ and
    $( s_1, s_2)\in R_2^j$:
    \begin{itemize}
    \item
      $\forall s_2\must{}_2 N_2: \exists s_1\must{}_1 N_1: \forall( a,
      t_1)\in N_1: \exists( a, t_2)\in N_2:( t_1, t_2)\in R_2^j$;
    \item if $j< k$, then
      $\forall s_1\may a_1 t_1: \exists s_2\may a_2 t_2:( t_1, t_2)\in
      R_2^{ j+ 1}$.
    \end{itemize}
  \item for all \emph{odd} $j\in\{ 0,\dotsc, k\}$ and
    $( s_1, s_2)\in R_2^j$:
    \begin{itemize}
    \item
      $\forall s_1\may a_1 t_1: \exists s_2\may a_2 t_2:( t_1, t_2)\in
      R_2^j$;
    \item if $j< k$, then
      $\forall s_2\must{}_2 N_2: \exists s_1\must{}_1 N_1: \forall( a,
      t_1)\in N_1: \exists( a, t_2)\in N_2:( t_1, t_2)\in R_2^{ j+
        1}$;
    \end{itemize}
  \end{itemize}
  A \emph{branching $k$-ready relation family} from $\mcal D_1$ to
  $\mcal D_2$ is a branching $k$-switching relation family as above
  with the extra property that if $k$ is even, then
  \begin{itemize}
  \item
    $\forall( s_1, s_2)\in R_1^k: \forall s_2\must{}_2 N_2: \exists
    s_1\must{}_1 N_1: \forall( a, t_1)\in N_1: \exists( a, t_2)\in
    N_2$;
  \item
    $\forall( s_1, s_2)\in R_2^k: \forall s_1\may a_1 t_1: \exists
    s_2\may a_2 t_2$;
  \end{itemize}
  and if $k$ is odd, then
  \begin{itemize}
  \item
    $\forall( s_1, s_2)\in R_1^k: \forall s_1\may a_1 t_1: \exists
    s_2\may a_2 t_2$;
  \item
    $\forall( s_1, s_2)\in R_2^k: \forall s_2\must{}_2 N_2: \exists
    s_1\must{}_1 N_1: \forall( a, t_1)\in N_1: \exists( a, t_2)\in
    N_2$.
  \end{itemize}
\end{definition}

For $k\ge 0$ and $\mcal D_1, \mcal D_2\in \DMTS$, we write
$\mcal D_1\le_k \mcal D_2$ if there exist a branching $k$-switching
relation family from $\mcal D_1$ to $\mcal D_2$.  We write
$\mcal D_1\le_k^{ \textup r} \mcal D_2$ if there exist a branching
$k$-ready relation family from $\mcal D_1$ to $\mcal D_2$.  Note that
$\le_0$ is the relation $\lesim$ from the preceding section.

\begin{theorem}
  \label{th:spec-game}
  For any $k\ge 0$, $( \DMTS, \chi, \le_k)$ is a specification theory
  for $\LTS$ adequate for $\sim_k$, and
  $( \DMTS, \chi, \le_k^{ \textup r})$ is a specification theory for
  $\LTS$ adequate for $\sim_k^{ \textup r}$.
\end{theorem}

\begin{proof}
Let $k\ge 0$.  We show that $( \DMTS, \chi, \le_k)$ is a specification
theory for $\LTS$ adequate for $\sim_k$; the proof for
$\le_k^{ \textup r}$ is similar.  We will apply Lemma~\ref{le:specth}.
Let
$\mcal I_1=( S_1, s^0_1, T_1), \mcal I_2=( S_2, s^0_2, T_2)\in \LTS$
and write $\chi( \mcal I_1)=( S_1,\{ s^0_1\}, \omay_1, \omust_1)$ and
$\chi( \mcal I_2)=( S_2,\{ s^0_2\}, \omay_2, \omust_2)$; we must prove
that $\chi( \mcal I_1)\le_k \chi( \mcal I_2)$ iff
$\mcal I_1\sim_k \mcal I_2$.

Assume that $\chi( \mcal I_1)\le_k \chi( \mcal I_2)$ and let
$R_1^0,\dotsc, R_1^k, R_2^0,\dotsc, R_2^k\subseteq S_1\times S_2$ be a
DMTS-branching $k$-switching relation family from $\chi( \mcal I_1)$
to $\chi( \mcal I_2)$ as of Definition~\ref{de:bkswitch-dmts}.  We show that
$R_1^0,\dotsc, R_1^k$ is an LTS-branching $k$-switching relation
family from $\mcal I_1$ to $\mcal I_2$ as of
Definition~\ref{de:bkswitch-lts}.
First, we have $( s_1^0, s_2^0)\in R_1^0$.

Let $j\in\{ 0,\dotsc, k\}$ even and $( s_1, s_2)\in R_1^j$.  Let
$( s_1, a, t_1)\in T_1$, then $s_1\may a_1 t_1$, hence there is
$s_2\may a_2 t_2$ such that $( t_1, t_2)\in R_1^j$, but then also
$( s_2, a, t_2)\in T_2$.  If $j< k$, then let $( s_2, a, t_2)\in T_2$,
thus $s_2\must{}_2 N_2=\{( a, t_2)\}$.  Hence there is
$s_1\must{}_1 N_1$ such that
$\forall( a, t_1)\in N_1: \exists( a, t_2')\in N_2:( t_1, t_2')\in
R_1^{ j+ 1}$.
But then $t_2'= t_2$ and $N_1=\{( a, t_1)\}$, hence
$( s_1, a, t_1)\in T_1$.  The arguments for $j$ odd are similar.


We have shown that $R_1^0,\dotsc, R_1^k$ is an LTS-branching
$k$-switching relation family from $\mcal I_1$ to $\mcal I_2$.
Analogously, one can show that $R_2^0,\dotsc, R_2^k$ is an
LTS-branching $k$-switching relation family from $\mcal I_2$ to
$\mcal I_1$.  The proof that $\mcal I_1\sim_k \mcal I_2$ implies
$\chi( \mcal I_1)\le_k \chi( \mcal I_2)$ proceeds along similar
lines. \qed
\end{proof}

\begin{remark}
  \label{re:simgame}
  There is a setting of \emph{generalized simulation games}, based on
  Stirling's bisimulation games~\cite{DBLP:conf/banff/Stirling95},
  which generalizes the above constructions and gives them a natural
  context.  We have developed these in a quantitative setting
  in~\cite{DBLP:journals/tcs/FahrenbergL14}, and we provide an
  exposition of the approach in Section~\ref{se:simgame}.  Generalized
  simulation games can be lifted to games on DMTS which can be used to
  define the relations of Definition~\ref{de:bkswitch-dmts}, see
  Section~\ref{se:spec-games}.
\end{remark}

\section{Specification Theories for Linear Equivalences}

We develop a scheme similar to the one of the previous section to
cover all linear relations in the linear-time--branching-time
spectrum.  For $\mcal I=( S, s^0, T)\in \LTS$, we let
$T^*\subseteq S\times \Sigma^*\times S$ be the reflexive, transitive
closure of $T$; a recursive definition is as follows:
\begin{itemize}
\item $( s, \epsilon, s)\in T^*$ for all $s\in S$;
\item for all $( s, \tau, t)\in T^*$ and $( t, a, u)\in T$, also $( s,
  \tau. a, u)\in T^*$.
\end{itemize}

\begin{definition}
  \label{de:lkswitch-lts}
  Let $k\ge 0$ and
  $\mcal I_1=( S_1, s^0_1, T_1), \mcal I_2=( S_2, s^0_2, T_2)\in
  \LTS$.
  A \emph{linear $k$-switching relation family} from $\mcal I_1$ to
  $\mcal I_2$ consists of relations
  $R^0\!,\dotsc, R^k\subseteq S_1\times S_2$ such that
  $( s_1^0, s_2^0)\in R^0$ and
  \begin{itemize}
  \item for all \emph{even} $j\in\{ 0,\dotsc, k\}$ and
    $( s_1, s_2)\in R^j$:
    \begin{itemize}
    \item
      $\forall( s_1, \tau, t_1)\in T_1^*: \exists( s_2, \tau, t_2)\in
      T_2^*$;
    \item if $j< k$, then
      $\forall( s_1, \tau, t_1)\in T_1^*: \exists( s_2, \tau, t_2)\in
      T_2^*:( t_1, t_2)\in R^{ j+ 1}$;
    \end{itemize}
  \item for all \emph{odd} $j\in\{ 0,\dotsc, k\}$ and
    $( s_1, s_2)\in R^j$:
    \begin{itemize}
    \item
      $\forall( s_2, \tau, t_2)\in T_2^*: \exists( s_1, \tau, t_1)\in
      T_1^*$;
    \item if $j< k$, then
      $\forall( s_2, \tau, t_2)\in T_2^*: \exists( s_1, \tau, t_1)\in
      T_1^*:( t_1, t_2)\in R^{ j+ 1}$;
    \end{itemize}
  \end{itemize}
\end{definition}

Hence a linear $0$-switching relation family is a \emph{trace
  inclusion}, and a linear $1$-switching relation family is a
\emph{impossible-futures inclusion}: any trace in $\mcal I_1$ has to
be matched by a trace in $\mcal I_2$, and then any trace from the end
of the second trace has to be matched by one from the end of the first
trace.

\begin{definition}
  \label{de:lkready-lts}
  Let $k\ge 0$ and
  $\mcal I_1=( S_1, s^0_1, T_1), \mcal I_2=( S_2, s^0_2, T_2)\in
  \LTS$.
  A \emph{linear $k$-ready relation family} from $\mcal I_1$ to
  $\mcal I_2$ is a linear $k$-switching relation family
  $R^0\!,\dotsc, R^k\subseteq S_1\times S_2$ with the extra property
  that for all $( s_1, s_2)\in R^k$:
  \begin{itemize}
  \item if $k$ is even, then
    $\forall( s_1, \tau, t_1)\in T_1^*: \exists( s_2, \tau, t_2)\in
    T_2^*: \forall( t_2, a, u_2)\in T_2: \exists( t_1, a,
    u_1)\in T_1$;
  \item if $k$ is odd, then
    $\forall( s_2, \tau, t_2)\in T_2^*: \exists( s_1, \tau, t_1)\in
    T_1^*: \forall( t_1, a, u_1)\in T_1: \exists( t_2, a, u_2)\in
    T_2$.
  \end{itemize}
\end{definition}

Thus a linear $0$-ready relation family is a \emph{failure inclusion}:
any trace in $\mcal I_1$ has to be matched by a trace in $\mcal I_2$
such that there is an inclusion of \emph{failure sets} of
non-available actions.  For $k\ge 0$ and
$\mcal I_1, \mcal I_2\in \LTS$, we write
$\mcal I_1\approx_k \mcal I_2$ if there exist a branching
$k$-switching relation family from $\mcal I_1$ to $\mcal I_2$ and
another from $\mcal I_2$ to $\mcal I_1$.  We write
$\mcal I_1\approx_k^{ \textup r} \mcal I_2$ if there exist a branching
$k$-ready relation family from $\mcal I_1$ to $\mcal I_2$ and another
from $\mcal I_2$ to $\mcal I_1$.

For $\mcal D=( S, S^0, \omay, \omust)\in \DMTS$, we define
$\mathord{ \smay{}}, \mathord{ \smust{}}\subseteq S\times
\Sigma^*\times S$ recursively as follows:
\begin{itemize}
\item $s\smay \epsilon s$ and $s\smust \epsilon s$ for all $s\in S$;
\item for all $s\smay \tau t$ and $t\may a u$, also
  $s\smay{ \tau. a} u$;
\item for all $s\smust \tau t$, $t\must{} N$, and $( a, u)\in N$, also
  $s\smust{ \tau. a} u$.
\end{itemize}

\begin{definition}
  \label{de:lkswitch-dmts}
  Let $k\ge 0$,
  $\mcal D_1=( S_1, S^0_1, \omay_1, \omust_1), \mcal D_2=( S_2, S^0_2,
  \omay_2, \omust_2)\in \DMTS$.  A \emph{linear $k$-switching relation
    family} from $\mcal D_1$ to $\mcal D_2$ consists of relations
  $R_1^0,\dotsc, R_1^k, R_2^0,\dotsc, R_2^k\subseteq S_1\times S_2$
  such that
  \begin{itemize}
  \item
    $\forall s^0_1\in S^0_1: \exists s^0_2\in S^0_2:( s^0_1, s^0_2)\in
    R_1^0$
    and
    $\forall s^0_2\in S^0_2: \exists s^0_1\in S^0_1:( s^0_1, s^0_2)\in
    R_2^0$;
  \item for all \emph{even} $j\in\{ 0,\dotsc, k\}$ and
    $( s_1, s_2)\in R_1^j$:
    \begin{itemize}
    \item
      $\forall s_1\smay \tau_1 t_1: \exists s_2\smay \tau_2 t_2$;
    \item if $j< k$, then
      $\forall s_1\smay \tau_1 t_1: \exists s_2\smay \tau_2 t_2:( t_1,
      t_2)\in R_1^{ j+ 1}$;
    \end{itemize}
  \item for all \emph{odd} $j\in\{ 0,\dotsc, k\}$ and
    $( s_1, s_2)\in R_1^j$:
    \begin{itemize}
    \item
      $\forall s_2\smust \tau_2 t_2: \exists s_1\smust \tau_1 t_1$;
    \item if $j< k$, then
      $\forall s_2\smust \tau_2 t_2: \exists s_1\smust \tau_1 t_1:(
      t_1, t_2)\in R_1^{ j+ 1}$;
    \end{itemize}
  \item for all \emph{even} $j\in\{ 0,\dotsc, k\}$ and
    $( s_1, s_2)\in R_2^j$:
    \begin{itemize}
    \item
      $\forall s_2\smust \tau_2 t_2: \exists s_1\smust \tau_1 t_1$;
    \item if $j< k$, then
      $\forall s_2\smust \tau_2 t_2: \exists s_1\smust \tau_1 t_1:(
      t_1, t_2)\in R_1^{ j+ 1}$;
    \end{itemize}
  \item for all \emph{odd} $j\in\{ 0,\dotsc, k\}$ and
    $( s_1, s_2)\in R_2^j$:
    \begin{itemize}
    \item
      $\forall s_1\smay \tau_1 t_1: \exists s_2\smay \tau_2 t_2$;
    \item if $j< k$, then
      $\forall s_1\smay \tau_1 t_1: \exists s_2\smay \tau_2 t_2:( t_1,
      t_2)\in R_2^{ j+ 1}$.
    \end{itemize}
  \end{itemize}
  A \emph{linear $k$-ready relation family} from $\mcal D_1$ to
  $\mcal D_2$ is a linear $k$-switching relation family as above with
  the extra property that if $k$ is even, then
  \begin{itemize}
  \item
    $\forall( s_1, s_2)\in R_1^k: \forall s_1\smay \tau_1 t_1: \exists
    s_2\smay \tau_2 t_2: \forall t_2\must{}_2 N_2: \exists t_1\must{}_1
    N_1: \forall( a, u_1)\in N_1: \exists( a, u_2)\in N_2$;
  \item
    $\forall( s_1, s_2)\in R_2^k: \forall s_2\smust \tau_2
    t_2: \exists s_1\smust \tau_1 t_1: \forall t_1\may a_1
    u_1: \exists t_2\may a_2 u_2$;
  \end{itemize}
  and if $k$ is odd, then
  \begin{itemize}
  \item
    $\forall( s_1, s_2)\in R_1^k: \forall s_2\smust \tau_2 t_2:
    \exists s_1\smust \tau_1 t_1: \forall t_1\may a_1 u_1: \exists
    t_2\may a_2 u_2$;
  \item
    $\forall( s_1, s_2)\in R_2^k: \forall s_1\smay \tau_1 t_1: \exists
    s_2\smay \tau_2 t_2: \forall t_2\must{}_2 N_2: \exists t_1\must{}_1
    N_1: \forall( a, u_1)\in N_1: \exists( a, u_2)\in N_2$;
  \end{itemize}
\end{definition}

For $k\ge 0$ and $\mcal D_1, \mcal D_2\in \DMTS$, we write
$\mcal D_1\lle_k \mcal D_2$ if there exists a linear $k$-switching
relation family from $\mcal D_1$ to $\mcal D_2$ and
$\mcal D_1\lle_k^{ \textup r} \mcal D_2$ if there exists a linear
$k$-ready relation family from $\mcal D_1$ to $\mcal D_2$.

\begin{theorem}
  \label{th:spec-lin}
  For any $k\ge 0$, $( \DMTS, \chi, \lle_k)$ is a specification theory
  for $\LTS$ adequate for $\approx_k$, and
  $( \DMTS, \chi, \lle_k^{ \textup r})$ is a specification theory for
  $\LTS$ adequate for $\approx_k^{ \textup r}$.
\end{theorem}

\begin{proof}
  Let $k\ge 0$.  We first show that $( \DMTS, \chi, \lle_k)$ is a
  specification theory for $\LTS$ adequate for $\approx_k$.  We will
  apply Lemma~\ref{le:specth}.

  Let
  $\mcal I_1=( S_1, s^0_1, T_1), \mcal I_2=( S_2, s^0_2, T_2)\in \LTS$
  and denote $\chi( \mcal I_1)=( S_1,\{ s^0_1\},$ $\omay_1, \omust_1)$
  and $\chi( \mcal I_2)=( S_2,\{ s^0_2\}, \omay_2, \omust_2)$.  We
  show that $\chi( \mcal I_1)\lle_k \chi( \mcal I_2)$ implies
  $\mcal I_1\approx_k \mcal I_2$; the other direction is similar.

Assume that $\chi( \mcal I_1)\lle_k \chi( \mcal I_2)$ and let
$R_1^0,\dotsc, R_1^k, R_2^0,\dotsc, R_2^k\subseteq S_1\times S_2$ be a
DMTS-linear $k$-switching relation family from $\chi( \mcal I_1)$ to
$\chi( \mcal I_2)$ as of Definition~\ref{de:lkswitch-dmts}.  We show that
$R_1^0,\dotsc, R_1^k$ is an LTS-linear $k$-switching relation family
from $\mcal I_1$ to $\mcal I_2$ as of Definition~\ref{de:lkswitch-lts}.
First, we have $( s_1^0, s_2^0)\in R_1^0$.

Let $j\in\{ 0,\dotsc, k\}$ even and $( s_1, s_2)\in R_1^j$.  Let
$( s_1, \tau, t_1)\in T_1^*$, then $s_1\smay \tau_1 t_1$, hence there
is $s_2\smay \tau_2 t_2$, implying that $( s_2, \tau, t_2)\in T_2^*$.
If $j< k$, then there is also $s_2\smay \tau_2 t_2$ such that
$( t_1, t_2)\in R_1^{ j+ 1}$, and again $( s_2, \tau, t_2)\in T_2^*$.

Let $j\in\{ 0,\dotsc, k\}$ odd and $( s_1, s_2)\in R_1^j$.  Let
$( s_2, \tau, t_2)\in T_2^*$, then $s_2\smust \tau_2 t_2$.  Hence
there is $s_1\smust \tau_1 t_1$, \ie~$( s_1, \tau, t_1)\in T_1^*$.  If
$j< k$, then there is $s_1\smust \tau_1 t_1$,
\ie~$( s_1, \tau, t_1)\in T_1^*$, such that
$( t_1, t_2)\in R_1^{ j+ 1}$.

We have shown that $R_1^0,\dotsc, R_1^k$ is an LTS-linear
$k$-switching relation family from $\mcal I_1$ to $\mcal I_2$.
Similarly, one can show that $R_2^0,\dotsc, R_2^k$ is an LTS-linear
$k$-switching relation family from $\mcal I_2$ to $\mcal I_1$.

Now assume that
$\chi( \mcal I_1)\lle_k^{ \textup r} \chi( \mcal I_2)$; we show that
$\mcal I_1\approx_k^{ \textup r} \mcal I_2$ (the other direction is
again similar).  Let
$R_1^0,\dotsc, R_1^k, R_2^0,\dotsc, R_2^k\subseteq S_1\times S_2$ be a
DMTS-linear $k$-ready relation family from $\chi( \mcal I_1)$ to
$\chi( \mcal I_2)$.  We show that $R_1^0,\dotsc, R_1^k$ is an
LTS-linear $k$-ready relation family from $\mcal I_1$ to $\mcal I_2$;
again, the proof that $R_2^0,\dotsc, R_2^k$ is an LTS-linear $k$-ready
relation family from $\mcal I_2$ to $\mcal I_1$ is completely
analogous.  First, we have $( s_1^0, s_2^0)\in R_1^0$.

We already know that $R_1^0,\dotsc, R_1^k$ is an LTS-linear
$k$-switching relation family from $\mcal I_1$ to $\mcal I_2$, so we
only need to see the extra conditions in Definition~\ref{de:lkready-lts}.
Let $( s_1, s_2)\in R_1^k$ and assume $k$ to be even (the proof is
similar for $k$ odd).  Let $( s_1, \tau, t_1)\in T_1^*$, then
$s_1\smay \tau_1 t_1$, hence there is $s_2\smay \tau_2 t_2$,
\ie~$( s_2, \tau, t_2)\in T_2^*$, such that
$\forall t_2\must{}_2 N_2: \exists t_1\must{}_1 N_1: \forall( a,
u_1)\in N_1: \exists( a, u_2)\in N_2$.

Let $( t_2, a, u_2)\in T_2$, then $t_2\must{}_2 N_2=\{( a, u_2)\}$.
Hence there is $t_1\must{}_1 N_1$ such that
$\forall( a, u_1)\in N_1: \exists( a, u_2')\in N_2$, but then
$N_1=\{( a, u_1)\}$, hence $( t_1, a, u_1)\in T_1$. \qed
\end{proof}


\section{Generalized Simulation Games}
\label{se:simgame}

In order to provide context to the constructions in
Sect.~\ref{se:bspec}, we introduce a notion of \emph{generalized
  simulation game}.  This is a generalization of Stirling's
bisimulation game~\cite{DBLP:conf/banff/Stirling95} which permits to
define most of the preorders and equivalences in van Glabbeek's
linear-time--branching-time spectrum~\cite{inbook/hpa/Glabbeek01}.
See also~%
\cite{DBLP:journals/tcs/FahrenbergL14} for a quantitative version of
these games.

Let
$\mcal I_1=( S_1, s^0_1, T_1), \mcal I_2=( S_2, s^0_2, T_2)\in \LTS$.
We will define a game played by two players, I and II, which
intuitively proceeds as follows.  Starting from the initial
configuration $( s^0_1, s^0_2)$, player~I chooses a transition from
$s^0_1$.  Player~II then has to match this with a transition with the
same label from $s^0_2$, and the game continues from the new
configuration $( s_1, s_2)$ given by the target states of the two
chosen transitions.  The game is won by player~I if she plays a
transition which player~II cannot match; if this never happens,
player~II wins.

We will see below that player~II has a strategy to always win this
game iff there is a \emph{simulation} from $\mcal I_1$ to $\mcal I_2$.
In order to characterize other preorders and equivalences, we
introduce some variability into the game:
\begin{itemize}
\item In any configuration $( s_1, s_2)$, player~I may choose to
  \emph{switch sides} and from now on play transitions from the right
  ($s_2$) component instead of the left, which player~II then has to
  answer by matching transitions on the left side.  Player~I may later
  choose to switch sides again.
\item In any configuration $( s_1, s_2)$, player~I may also choose to
  play a \emph{last} transition which ends the game.  If player~II can
  match the transition, then she has won; otherwise, player~I wins.
\end{itemize}
Different combinations of these variations, together with restrictions
on when and how often player~I is allowed to switch sides, will define
games which characterize all branching equivalences in the
linear-time--branching-time spectrum.

We formalize the above description.  The sets of \emph{extended
  states} for the players are
\begin{align*}
  C_1 &= ( T_1\times T_2\cup T_2\times T_1)^*\,, \\
  C_2 &= ( T_1\times T_2\cup T_2\times T_1)^*\!.( T_1\cup T_2)\,.
\end{align*}
These keep track of which edges have been previously chosen by the
players.  Note that $C_1$ contains the empty extended state
$\epsilon$.

A \emph{strategy for player~I} is a partial mapping
$\theta_1: C_1\parto T_1\cup T_2$ such that whenever
$\theta_1((( s_1, a_1, t_1),( s_1', a_1', t_1'))\dotsc(( s_n, a_n,
t_n),( s_n', a_n', t_n')))=( s, a, t)$
is defined, then $s= t_n$ or $s= t_n'$.  Hence an edge chosen by
player~I must extend one of the previous two edges.  If
$\theta_1( \epsilon)=( s, a, t)$ is defined, then $s= s^0_1$ or
$s= s^0_2$.  The set of strategies for player~I is denoted $\Theta_1$.
For $c_1\in C_1$ and $\theta_1\in \Theta_1$, the \emph{update}
$\upd( c_1)$ of $c_1$ is defined iff $\theta_1( c_1)$ is defined, and
then $\upd( c_1)= c_1. \theta_1( c_1)\in C_2$.

A \emph{strategy for player~II} is a partial mapping
$\theta_2: C_2\parto T_1\cup T_2$ such that whenever
$\theta_1((( s_1, a_1, t_1),( s_1', a_1', t_1'))\dotsc(( s_n, a_n,
t_n),( s_n', a_n', t_n')).( s, a, t))=( s', a', t')$ is defined, then
$a= a'$, and
\begin{itemize}
\item if $s= t_n$, then $( s', a', t')\in T_2$ and $s'= t_n'$;
\item if $s= t_n'$, then $( s', a', t')\in T_1$ and $s'= t_n$.
\end{itemize}
Hence player~II has to play a transition with the same label as the
last transition played by player~I and on the opposite side of the
game.  The set of strategies for player~II is denoted $\Theta_2$.  For
$c_2\in C_2$ and $\theta_2\in \Theta_2$, the \emph{update}
$\upd_2( c_2)$ of $c_2$ is defined iff $\theta_2( c_2)$ is defined,
and then $\upd_2( c_2)= c_2. \theta_2( c_2)\in C_1$.

Now let $( \theta_1, \theta_2)\in \Theta_1\times \Theta_2$ be a
\emph{strategy pair}, then this induces a finite or infinite
alternating sequence $( c_1^0, c_2^1, c_1^1, c_2^2,\dotsc)$ of
extended states, where $c_1^0= \epsilon$ and for all
$j\ge 1$,
\begin{itemize}
\item $c_2^j$ is defined iff $\theta_2( c_1^{ j- 1})$ is defined, and
  then $c_2^j= \theta_2( c_1^{ j- 1})$;
\item $c_1^j$ is defined iff $\theta_1( c_2^j)$ is defined, and then
  $c_1^j= \theta_1( c_2^j)$.
\end{itemize}
Each extended state in the sequence is a prefix of the succeeding one,
hence these define a finite or infinite string
\begin{equation*}
  \sigma( \theta_1, \theta_2)\in C_1\cup C_2\cup( T_1\times T_2\cup
  T_2\times T_1)^\omega\,.
\end{equation*}

A strategy $\theta_1\in \Theta_1$ is \emph{winning for player~I} if
$\sigma( \theta_1, \theta_2)\in C_2$ for all $\theta_2\in \Theta_2$.
A strategy $\theta_2\in \Theta_2$ is \emph{winning for player~II} if
$\sigma( \theta_1, \theta_2)\in C_1\cup( T_1\times T_2\cup T_2\times
T_1)^\omega$
for all $\theta_1\in \Theta_1$.  The game is determined, so that
player~I has a winning strategy iff player~II does not.

\begin{remark}
  \label{re:memoryless-gensim}
  As the game is about player~II matching transitions played by
  player~I, and once she has done so, past transition labels are
  ignored, it is clear that it suffices to consider \emph{memory-less}
  strategies for both players, \ie~strategies where the transitions
  chosen only depend on the current game configuration instead of all
  past moves.  This is important from an algorithmic point of view,
  but we will not need it below.
\end{remark}

We introduce a \emph{switch counter} $\sco$ which indicates how often
player~I has switched sides to arrive at a given extended state
$c_1\in C_1=( T_1\times T_2\cup T_2\times T_1)^*$.  Intuitively,
$\sco( c_1)$ counts how often the elements in the sequence $c_1$
switch from being in $T_1\times T_2$ to being in $T_2\times T_1$ and
vice versa.  Hence $\sco( c_1)= 0$ iff
$c_1\in( T_1\times T_2)^*\cup( T_2\times T_1)^*$, $\sco( c_1)= 1$ iff
$c_1\in( T_1\times T_2)^+( T_2\times T_1)^+\cup( T_2\times T_1)^+(
T_1\times T_2)^+$,
etc.  For $c_2\in C_2$, we similarly have $\sco( c_2)= 0$ iff
$c_2\in( T_1\times T_2)^* T_1\cup( T_2\times T_1)^* T_2$,
$\sco( c_2)= 1$ iff
$c_2\in( T_1\times T_2)^+ T_2\cup( T_2\times T_1)^+ T_1$, etc.

\begin{definition}
  \label{de:kswitch}
  Let $k\ge 0$.  A strategy $\theta_1\in \Theta_1$ is
  \emph{$k$-switching} if $\sco( \theta_1( c_1))\le k$ for all
  $c_1\in C_1$ for which $\theta( c_1)$ is defined.  It is
  \emph{$k$-ready switching} if $\sco( c_1)\le k$ for all $c_1\in C_1$
  for which $\theta( c_1)$ is defined.
\end{definition}

Hence a $0$-switching strategy for player~I can never switch sides,
whence a $0$-ready switching strategy can switch sides once, but must
be undefined after.  Similarly, a $1$-switching strategy can switch
sides once, and a $1$-ready switching strategy can then switch once
more, but no more player~I moves are defined after.  We denote the
sets of $k$-switching strategies by $\Theta_1^k$ and of $k$-ready
switching strategies by $\Theta_1^{ k\textup{-r}}$.  Note that
$\Theta_1^k\subseteq \Theta_1^{ k\textup{-r}}$ for all $k\ge 0$, and
$\Theta_1^\infty= \Theta_1^{ \infty\textup{-r}}= \Theta_1$.

For any subset $\Theta_1'\subseteq \Theta_1$, the
\emph{$\Theta_1'$-game} denotes the above game when player~I is only
permitted to use strategies in $\Theta_1'$.

\begin{proposition}
  \label{pr:simgame-app}
  Let $k\ge 0$ and $\mcal I_1, \mcal I_2\in \LTS$.  Then
  $\mcal I_1\sim_k \mcal I_2$ iff player~II has a winning strategy in
  the $\Theta_1^k$-game on $\mcal I_1, \mcal I_2$, and
  $\mcal I_1\sim_k^{ \textup r} \mcal I_2$ iff player~II has a winning
  strategy in the $\Theta_1^{ k\textup{-r}}$-game on
  $\mcal I_1, \mcal I_2$.
\end{proposition}

\begin{proof}
  If $\theta_2\in \Theta_2$ is winning for player~II in the
  specification $\Theta_1^k$-game, then any strategy pair
  $( \theta_1, \theta_2)$ can be used to construct a branching
  $k$-switching relation family.  Conversely, any branching
  $k$-switching relation family can be used to construct a
  (memory-less) winning player-II strategy in the specification
  $\Theta_1^k$-game.  The proof is similar for the $k$-ready case. \qed
\end{proof}



\begin{remark}
  \label{re:ltbt-game}
  By suitably modifying the $\sco$ notion, also \emph{preorders} in
  the spectrum can be characterized.  By introducing a notion of
  \emph{blind} strategy for player~I, also linear relations in the
  spectrum can be covered.  See~\cite{DBLP:journals/tcs/FahrenbergL14}
  for details.
\end{remark}

\section{Specification Games}
\label{se:spec-games}

We can now use the developments in the last section to introduce
general specification games on $\DMTS$ which can be instantiated to
yield specification theories which are adequate for any equivalence in
the linear-time--branching-time spectrum.

Let
$\mcal D_1=( S_1, S^0_1, \omay_1, \omust_1), \mcal D_2=( S_2, S^0_2,
\omay_2, \omust_2)\in \DMTS$.
The sets of \emph{extended states} for the players are
\begin{align*}
  C_1 &= (( \omay_1\times \omay_2)\cup( \omust_2\times \omust_1\times
  \Sigma\times S_1\times \Sigma\times S_2))^*\,, \\
  C_2 &= (( \omay_1\times \omay_2)\cup( \omust_2\times \omust_1\times
  \Sigma\times S_1\times \Sigma\times S_2))^*\!.( \omay_1\cup
  \omust_2)\,, \\
  C_1' &= (( \omay_1\times \omay_2)\cup( \omust_2\times \omust_1\times
  \Sigma\times S_1\times \Sigma\times S_2))^*\!.( \omust_2\times
  \omust_1)\,, \\
  C_2' &= (( \omay_1\times \omay_2)\cup( \omust_2\times \omust_1\times
  \Sigma\times S_1\times \Sigma\times S_2))^*\!. \\
  &\hspace*{20em} ( \omust_2\times
  \omust_1\times \Sigma\times S_1)\,.
\end{align*}
This conveys the following intuition: At each round of the game,
player~I either plays a may-transition in $\mcal D_1$ or a disjunctive
must-transition in $\mcal D_2$.  In the first case, player~II answers
with a matching may-transition in $\mcal D_2$, and the game proceeds.
In the second case, player~II answers with a disjunctive
must-transition in $\mcal D_1$, bringing the game into a state where
player~I now must play a branch $( a, t)$ of the chosen
must-transition in $\mcal D_1$.  To this, player~II must answer with a
matching branch in the must-transition in $\mcal D_2$, and then the
game can proceed.

A \emph{strategy for player~I} hence consists of two partial mappings
$\theta_1: C_1\parto( \omay_1\cup \omust_2)$, $\theta_1': C_1'\parto
\Sigma\times S_1$ such that
\begin{itemize}
\item if $c_1= c_1^1\dotsc c_1^n\in C_1$, $\theta_1( c_1)$ is defined,
  and
  $c_1^n=(( s_n, a_n, t_n),( s_n', a_n', t_n'))\in( \omay_1\times
  \omay_2)$ or
  $c_1^n=(( s_n, N_n),( s_n', N_n'), a_n, t_n, a_n', t_n')\in(
  \omust_2\times \omust_1\times \Sigma\times S_1\times \Sigma\times
  S_2)$, then
  \begin{itemize}
  \item if $\theta_1( c_1)=( s, a, t)\in \omay_1$, then $s= t_n$;
  \item if $\theta_1( c_1)=( s, N)\in \omust_2$, then $s= t_n'$;
  \end{itemize}
\item if $c_1'= c_1''.(( s, N),( s', N'))\in C_1'$ and
  $\theta_1'( c_1')=( a, t)$ is defined, then $( a, t)\in N'$.
\end{itemize}
This says that from an extended state in $C_1$, player~I must choose
a transition from one of the previous target states, and from a state
in $C_1'$, player~I must choose a branch of the must-transition just
chosen by player~II.

If $\theta_1( \epsilon)$ is defined, then
\begin{itemize}
\item if $\theta_1( \epsilon)=( s, a, t)\in \omay_1$, then $s\in S_1^0$;
\item if $\theta_1( \epsilon)=( s, N)\in \omust_2$, then $s\in S_2^0$.
\end{itemize}

A \emph{strategy for player~II} consists of two partial mappings
$\theta_2: C_2\parto( \omay_2\cup \omust_1)$, $\theta_2': C_2'\parto
\Sigma\times S_2$ such that
\begin{itemize}
\item if $c_2= c_2^1\dotsc c_2^n. \tau\in C_2$ and $\theta_2( c_2)$ is
  defined, and
  $c_2^n=(( s_n, a_n, t_n),$ $( s_n', a_n', t_n'))\in( \omay_1\times
  \omay_2)$
  or
  $c_2^n=(( s_n, N_n),( s_n', N_n'), a_n, t_n, a_n', t_n')\in(
  \omust_2\times \omust_1\times \Sigma\times S_1\times \Sigma\times
  S_2)$, then
  \begin{itemize}
  \item if $\tau=( s, a, t)\in \omay_1$, then $\theta_2( c_2)=( s',
    a, t')\in \omay_2$ with $s'= t_n'$;
  \item if $\tau=( s, N)\in \omust_2$, then $\theta_2( c_2)=( s',
    N')\in \omust_1$ with $s'= t_n$;
  \end{itemize}
\item if $c_2'= c_2''.(( s, N),( s', N'),( a, t))\in C_2'$ and
  $\theta_2'( c_2')=( a', t')$ is defined, then $( a', t')\in N$ and
  $a'= a$.
\end{itemize}
The sets of strategies for players~I and~II are denoted $\Theta_1$ and
$\Theta_2$.

Let $( \theta_1, \theta_1')\in \Theta_1$,
$( \theta_2, \theta_2')\in \Theta_2$, $c_1\in C_1$, $c_2\in C_2$,
$c_1'\in C_1'$, and $c_2'\in C_2'$.  We define the update functions:
\begin{itemize}
\item If $\theta_1( c_1)$ is defined, then
  $\upd( c_1)= c_1. \theta_1( c_1)\in C_2$.
\item If $\theta_2( c_2)$ is defined, then
  $\upd( c_2)= c_2. \theta_2( c_2)\in C_1$ if
  $\theta_2( c_2)\in \omay_2$ and
  $\upd( c_2)= c_2. \theta_2( c_2)\in C_1'$ if
  $\theta_2( c_2)\in \omust_1$.
\item If $\theta_1'( c_1')$ is defined, then $\upd( c_1')=
  c_1'. \theta_1'( c_1')\in C_2'$.
\item If $\theta_2'( c_2')$ is defined, then $\upd( c_2')=
  c_2'. \theta_2'( c_2')\in C_1$.
\end{itemize}
Hence a strategy pair $(( \theta_1, \theta_1'),( \theta_2,
\theta_2'))\in \Theta_1\times \Theta_2$ induces, via the update
functions, a finite or infinite string
\begin{multline*}
  \sigma(( \theta_1, \theta_1'),( \theta_2, \theta_2'))\in C_1\cup
  C_2\cup C_1'\cup C_2' \\
  \cup(( \omay_1\times \omay_2)\cup( \omust_2\times \omust_1\times
  \Sigma\times S_1\times \Sigma\times S_2))^\omega\,. \hspace*{4em}
\end{multline*}

Then $( \theta_1, \theta_1')\in \Theta_1$ is said to be \emph{winning
  for player~I} if
$\sigma(( \theta_1, \theta_1'),( \theta_2, \theta_2'))\in C_2\cup
C_2'$
for all $( \theta_2, \theta_2')\in \Theta_2$, and
$( \theta_2, \theta_2')\in \Theta_2$ is \emph{winning for player~II}
if
$\sigma(( \theta_1, \theta_1'),( \theta_2, \theta_2'))\in C_1\cup
C_1'\cup(( \omay_1\times \omay_2)\cup( \omust_2\times \omust_1\times
\Sigma\times S_1\times \Sigma\times S_2))^\omega$
for all $( \theta_1, \theta_1')\in \Theta_1$.  The game is determined,
\ie~player~I has a winning strategy iff player~II does not.

We introduce a switching counter $\sco$, similarly to the one of the
preceding section.  For $c_1\in C_1$,
\begin{itemize}
\item $\sco( c_1)= 0$ iff
  $c_1\in( \omay_1\times \omay_2)^*\cup( \omust_2\times \omust_1\times
  \Sigma\times S_1\times \Sigma\times S_2)^*$;
\item $\sco( c_1)= 1$ iff
  $c_1\in( \omay_1\times \omay_2)^+( \omust_2\times
  \omust_1\times \Sigma\times S_1\times \Sigma\times S_2)^+\cup(
  \omust_2\times \omust_1\times \Sigma\times S_1\times \Sigma\times
  S_2)^+( \omay_1\times \omay_2)^+$;
\end{itemize}
etc., and for $c= c_1. c'\in C_2\cup C_1'\cup C_2'$ such that
$c_1\in C_1$ is the longest $C_1$-prefix of $c$,
$\sco( c)= \sco( c_1)$.  We also copy Definition~\ref{de:kswitch} to
introduce $k$-switching and $k$-ready switching strategies in
$\Theta_1$, and denote again the subsets of $k$-switching strategies
by $\Theta_1^k$ and of $k$-ready switching strategies by
$\Theta_1^{ k\textup{-r}}$.

\begin{proposition}
  Let $k\ge 0$ and $\mcal D_1, \mcal D_2\in \DMTS$.  Then
  $\mcal D_1\le_k \mcal D_2$ iff player~II has a winning strategy in
  the $\Theta_1^k$-game on $\mcal D_1, \mcal D_2$, and
  $\mcal D_1\le_k^{ \textup r} \mcal D_2$ iff player~II has a winning
  strategy in the $\Theta_1^{ k\textup{-r}}$-game on
  $\mcal D_1, \mcal D_2$.
\end{proposition}

\begin{proof}
  Similar to the proof of Proposition~\ref{pr:simgame-app}. \qed
\end{proof}

It is again sufficient to consider memory-less strategies for both
players, \cf~Remark~\ref{re:memoryless-gensim}.

\section{Game-Based Proof of Theorem~\ref{th:spec-game}}

We now show a game-based proof of Theorem~\ref{th:spec-game} which
relates $\le_k$ with $\sim_k$ and $\le_k^{ \textup r}$ with
$\sim_k^{ \textup r}$.  This is based on exposing an isomorphism between
generalized simulation games on LTS and corresponding specification
games on their embeddings into $\DMTS$.  Hence it can be used to show
the more general result that any restriction
$\Theta_1'\subseteq \Theta_1$ in the specification game yields a
specification theory adequate for an equivalence relation defined on
$\LTS$ by a similar restriction of the generalized simulation game.

\begin{proof}[of Theorem~\ref{th:spec-game}]
  We show that for $\mcal I_1, \mcal I_2\in \LTS$,
  $\chi( \mcal I_1)\le_k \chi( \mcal I_2)$ iff
  $\mcal I_1\sim_k \mcal I_2$ and apply Lemma~\ref{le:specth}; the
  proof for the $k$-ready relations is similar.

  The essence of the proof is that the simulation $\Theta_k$-game on
  $\mcal I_1$, $\mcal I_2$ and the specification $\Theta_k$-game on
  $\chi( \mcal I_1)$, $\chi( \mcal I_2)$ are isomorphic.  We expose an
  injective mapping $\Phi$, from extended states in the simulation game
  to extended states in the specification game, which essentially maps
  transitions in $\mcal I_1$ to may-transitions in $\chi( \mcal I_1)$
  and transitions in $\mcal I_2$ to must-transitions in $\chi( \mcal
  I_2)$.

  We then show that extended states outside the image of $\Phi$ are
  unreachable in any specification game, hence $\Phi$ is a bijection
  between extended states in the simulation game and ``proper''
  extended states in the specification game.

  Using this, we then extend $\Phi$ to an injective mapping from
  strategies in the simulation game to strategies in the specification
  game, and we show that strategies outside the image of $\Phi$ need
  not be considered.  Also, $\Phi$ preserves and reflects the
  switching counter, and we show that a strategy $\theta_1$ is winning
  for player~I in the simulation game iff $\Phi( \theta_1)$ is winning
  for player~I in the specification game.

  Write $\mcal I_1=( S_1, s^0_1, T_1)$,
  $\mcal I_2=( S_2, s^0_2, T_2)$,
  $\chi( \mcal I_1)=( S_1,\{ s^0_1\}, \omay_1, \omust_1)$, and
  $\chi( \mcal I_2)=( S_2,\{ s^0_2\}, \omay_2, \omust_2)$.  In this
  proof, we denote extended states and strategies in the specification
  game as in Sect.~\ref{se:spec-games}, whereas extended states and
  strategies in the game of Sect.~\ref{se:simgame} are denoted using
  tildes.

  We define mappings $\Phi_1: \tilde C_1\to C_1$,
  $\Phi_2: \tilde C_2\to C_2$.  Let
  $\phi_1:( T_1\times T_2\cup T_2\times T_1)\to(( \omay_1\times
  \omay_2)\cup( \omust_2\times \omust_1\times \Sigma\times S_1\times
  \Sigma\times S_2))$ and
  $\phi_2:( T_1\cup T_2)\to( \omay_1\cup \omust_2)$ be given by
  \begin{align*}
    \phi_1(( s, a, t),( s', a', t')) &=
    \begin{cases}
      (( s, a, t),( s', a', t')) &\text{if }( s, a, t)\in T_1\,, \\
      (( s,\{( a, t)\}),( s',\{( a', t')\}), a', t', a, t) &\text{if
      }( s, a, t)\in T_2\,,
    \end{cases} \\
    \phi_2( s, a, t) &=
    \begin{cases}
      ( s, a, t) &\text{if }( s, a, t)\in T_1\,, \\
      ( s,\{( a, t)\}) &\text{if }( s, a, t)\in T_2\,,
    \end{cases}
  \end{align*}
  and for $\tilde c_1= \tilde c_1^1\dotsc \tilde c_1^n\in \tilde C_1$
  and $\tilde c_2= \tilde c_1. \tilde c_2'\in \tilde C_2$, define
  $\Phi_1( \tilde c_1)= \phi_1( \tilde c_1^1)\dotsc \phi_1( \tilde
  c_1^n)$
  and
  $\Phi_2( \tilde c_2)= \Phi_1( \tilde c_1). \phi_2( \tilde c_2')$.
  We also define
  $\Phi_3:( T_1\times T_2\cup T_2\times T_1)^\omega\to( \omust_2\times
  \omust_1\times \Sigma\times S_1\times \Sigma\times S_2)^\omega$ by
  $\Phi_3( d_1 d_2\dotsc)= \phi_1( d_1) \phi_1( d_2)\dotsc$ and let
  $\Phi= \Phi_1\cup \Phi_2\cup \Phi_3$.

  We call extended states in the image of $\Phi_1$, $\Phi_2$
  \emph{proper}, and we note that any reachable extended state in $C_1$
  and $C_2$ is proper: Let $c_1= c_1^1\dotsc c_1^n\in C_1$ and
  $j\in\{ 1,\dotsc, n\}$ such that
  $c_1^j=(( s_j, N_j),( s_j', N_j'), a_j, t_j, a_j', t_j')\in(
  \omust_2\times \omust_1\times \Sigma\times S_1\times \Sigma\times
  S_2)$.
  Then $N_j=\{( b_j, u_j)\}$ and $N_j'=\{( b_j', u_j')\}$ for some
  $( s_j, b_j, u_j)\in T_2$ and $( s_j', b_j', u_j')\in T_1$.  Now if
  the extended state $c_1$ will be reached during any game, then
  $c_1^1\dotsc c_1^{ j- 1}.(( s_j, N_j),( s_j', N_j'))\in C_1'$ must
  also have been reached, and then $( a_j, t_j)\in N_j'$ and
  $( a_j', t_j')\in N_j$ by the definition of strategies.  But $N_j'$
  and $N_j$ are one-element sets, so that we must have $a_j= b_j'$,
  $t_j= u_j'$, $a_j'= b_j$, and $t_j'= u_j$.  Hence we can assume that
  if
  $c_1^j\in( \omust_2\times \omust_1\times \Sigma\times S_1\times
  \Sigma\times S_2)$,
  then
  $c_1^j=(( s_j,\{( b_j, u_j)\}),( s_j',\{( b_j', u_j')\}), b_j',
  u_j', b_j, u_j)$
  for some $( s_j, b_j, u_j)\in T_2$ and $( s_j', b_j', u_j')\in T_1$,
  \ie~$c_1^j= \phi_1(( s_j, b_j, u_j),( s_j', b_j', u_j'))$.

  The functions $\Phi_1$ and $\Phi_2$ are also injective, hence they
  are bijections onto the proper subsets of $C_1$ and $C_2$.  We have
  shown that improper extended states are not reachable, hence
  strategies in $\Theta_1$ and $\Theta_2$ need not be defined on
  improper extended states.

  Next we note that strategies $\theta_1': C_1'\to \Sigma\times S_1$
  and $\theta_2': C_2'\to \Sigma\times S_2$ are unique: If
  $c_1'= c_1''.(( s, N),( s', N'))\in C_1'$ and
  $\theta_1'( c_1')=( a, t)$ is defined, then $( a, t)\in N'$, but
  $N'=\{( b', u')\}$ is a one-element set, hence $a= b'$ and $t= u'$.
  If $\theta_1'( c_1')$ is undefined, then the modification of
  $\theta_1'$ which defines $\theta_1'( c_1')=( b', u')$ is better for
  player~I.  The argument is similar for player~II.  We can henceforth
  assume that $\theta_1'$ and $\theta_2'$ always are the strategies
  defined above.

  We extend the mappings $\Phi_1$ and $\Phi_2$ to strategies.  Let
  $\tilde \theta_1\in \tilde \Theta_1$, then
  $\Phi_1( \tilde \theta_1)=( \theta_1, \theta_1')$, where $\theta_1'$
  is the unique strategy as above,
  $\theta_1( c_1)= \phi_2( \tilde \theta_1( \Phi_1^{ -1}( c_1)))$ for
  any proper extended state $c_1\in C_1$, and $\theta_1( c_1)$
  undefined for $c_1$ improper.  Similarly, for
  $\tilde \theta_2\in \tilde \Theta_2$,
  $\Phi_2( \tilde \theta_2)=( \theta_2, \theta_2')$, where $\theta_2'$
  is the unique player-II strategy,
  $\theta_2( c_2)= \phi_2( \tilde \theta_2( \Phi_2^{ -1}( c_2)))$ for
  any proper extended state $c_2\in C_2$, and $\theta_2( c_2)$
  undefined for $c_2$ improper.  The so-defined functions
  $\Phi_1: \tilde \Theta_1\to \Theta_1$,
  $\Phi_2: \tilde \Theta_2\to \Theta_2$ are injective, hence
  bijections onto their images, which consist precisely of the
  strategies which are the unique strategies on $C_1'$ and $C_2'$ and
  undefined on improper extended states in $C_1$ and $C_2$.  $\Phi_1$
  also preserves and reflects switching counters: for all
  $\theta_1\in \Theta_1$ and $k\ge 0$,
  $\tilde \theta_1\in \tilde \Theta_1^k$ iff
  $\Phi_1( \tilde \theta_1)\in \Theta_1^k$ and
  $\tilde \theta_1\in \tilde \Theta_1^{ k\textup{-r}}$ iff
  $\Phi_1( \tilde \theta_1)\in \Theta_1^{ k\textup{-r}}$.

  Let $( \tilde \theta_1, \tilde \theta_2)\in \tilde \Theta_1\times
  \tilde \Theta_2$ be a strategy pair; we will show that $\sigma(
  \Phi_1( \tilde \theta_1), \Phi_2( \tilde \theta_2))= \Phi( \tilde
  \sigma( \tilde \theta_1, \tilde \theta_2))$.  Let $\tilde c_1\in
  \tilde C_1$, then
  \begin{multline*}
    \Phi_2( \upd( \tilde c_1))= \Phi_2( \tilde c_1. \tilde \theta_1(
    \tilde c_1))= \Phi_1( \tilde c_1). \phi_2( \tilde \theta_1( \tilde
    c_1)) \\= \Phi_1( \tilde c_1). \phi_2( \tilde \theta_1( \Phi_1^{ -1}(
    \Phi_1( \tilde c_1))))= \Phi_1( \tilde c_1). \theta_1( \Phi_1(
    \tilde c_1))= \upd( \Phi_1( \tilde c_1))\,,
  \end{multline*}
  where $\Phi_1( \tilde \theta_1)=( \theta_1, \theta_1')$.  This shows
  that $\Phi$ commutes with the update functions on $\tilde C_1$ and
  $C_1$.  Similarly one can show that $\Phi$ commutes with the update
  functions on $\tilde C_2$ and $C_2$, and the updates on $C_1'$ and
  $C_2'$ are unique because $\theta_1'$ and $\theta_2'$ are the unique
  strategies.  Together with $\Phi( \epsilon)= \epsilon$ and by
  induction, this implies that
  $\Phi( \tilde \sigma( \tilde \theta_1, \tilde \theta_2))= \sigma(
  \Phi_1( \tilde \theta_1), \Phi_2( \tilde \theta_2))$.

  We can now finish the proof.  Let $k\ge 0$ and assume
  $\chi( \mcal I_1)\not\le_k \chi( \mcal I_2)$, then player~I has a
  winning strategy $( \theta_1, \theta_1')\in \Theta_1^k$ in the
  specification $\Theta_1^k$-game on $\chi( \mcal I_1)$,
  $\chi( \mcal I_2)$.  We can assume that $( \theta_1, \theta_1')$ is
  in the image of $\Phi_1$, hence there is
  $\tilde \theta_1\in \tilde \Theta_1^k$ such that
  $\Phi_1( \tilde \theta_1)=( \theta_1, \theta_1')$.  We show that
  $\tilde \theta_1$ is winning for player~I in the
  $\tilde \Theta_1^k$-game on $\mcal I_1$, $\mcal I_2$, which will
  imply $\mcal I_1\not\sim_k \mcal I_2$.  Let
  $\tilde \theta_2\in \tilde \Theta_2$, then
  \begin{equation*}
    \tilde \sigma( \tilde \theta_1, \tilde \theta_2)= \Phi^{ -1}(
    \sigma( \Phi_1( \tilde \theta_1), \Phi_2( \tilde \theta_2)))\in
    \Phi^{ -1}( C_2)\subseteq \tilde C_2\,.
  \end{equation*}

  Now assume that $\mcal I_1\not\sim_k \mcal I_2$ and let
  $\tilde \theta_1\in \tilde \Theta_1^k$ be a winning strategy for
  player~I in the $\tilde \Theta_1^k$-game on $\mcal I_1$,
  $\mcal I_2$.  Let $( \theta_1, \theta_1')= \Phi( \tilde \theta_1)$,
  we show that $( \theta_1, \theta_1')$ is winning for player~I in the
  $\Theta_1^k$-game on $\chi( \mcal I_1)$, $\chi( \mcal I_2)$.  Let
  $( \theta_2, \theta_2')\in \Theta_2$, then we can assume that there
  is $\tilde \theta_2\in \tilde \Theta_2$ such that
  $\Phi_2( \tilde \theta_2)=( \theta_2, \theta_2')$, and
  \begin{equation*}
    \sigma(( \theta_1, \theta_1'),( \theta_2, \theta_2'))= \Phi(
    \tilde \sigma( \tilde \theta_1, \tilde \theta_2)\in \Phi( \tilde
    C_2)\subseteq C_2\,,
  \end{equation*}
  hence $\chi( \mcal I_1)\not\le_k \chi( \mcal I_2)$. \qed
\end{proof}

\section{Conclusion}

We have in this paper extracted a reasonable and general notion of
(behavioral) specification theory, based on previous work by a number
of authors on concrete specification theories in different contexts
and on the well-established notions of characteristic formulae,
adequacy and expressivity.

Using this general concept of specification theory, we have introduced
new concrete specification theories, based on disjunctive modal
transition systems, for most equivalences in van~Glabbeek's
linear-time--branching-time spectrum.  Previously, only specification
theories for bisimilarity have been available, and recent work by
Vogler~\etal calls for work on specification theories for failure
equivalence.  Both failure equivalence and bisimilarity are part of
the linear-time--branching-time spectrum, as are nested simulation
equivalence, impossible-futures equivalence, and many other useful
relations.  We develop specification theories for all branching
equivalences in the spectrum, but we miss some of the linear
equivalences; notably, possible futures and ready trace equivalence
are missing.  We believe that these can be captured by small
modifications to our setting, but leave this for future work.

Our new specification theories should be useful for example in the
setting of the failure semantics of Vogler~\etal, but also in many
other contexts where bisimilarity is not the right equivalence to
consider.  Using our own previous work on the quantitative
linear-time--branching-time spectrum and on quantitative specification
theories for bisimilarity, we also plan to lift our work presented
here to the quantitative setting.

Specification theories for bisimilarity admit notions of conjunction
and composition which enable compositional design and verification,
and also the specification theories of Vogler \etal have (different)
such notions.  Using the game-based setting, we believe one can define
general notions of conjunction and composition defined by games played
on the involved disjunctive modal transition systems.  This is left
for future work.

\bibliographystyle{plain}
\bibliography{mybib}

\end{document}